\documentclass[acmsmall]{ec21acm}

\makeatletter
\def\genfootnote{\gdef\@thefnmark{}\@footnotetext}
\let\@authorsaddresses\@empty
\makeatother

\AtBeginDocument{%
  \providecommand\BibTeX{{%
    \normalfont B\kern-0.5em{\scshape i\kern-0.25em b}\kern-0.8em\TeX}}}

\usepackage{tikz}
\usetikzlibrary{arrows,shapes}
\usepackage[gen]{eurosym}

\begin{document}

\title{Debt Swapping for Risk Mitigation in Financial Networks}

\author{Pál András Papp}
\email{apapp@ethz.ch}
\affiliation{%
  \institution{ ETH Zürich}
  \country{apapp@ethz.ch}
  \vspace{3pt}
}

\author{Roger Wattenhofer}
\email{wattenhofer@ethz.ch}
\affiliation{%
  \institution{ETH Zürich}
  \country{wattenhofer@ethz.ch}
}

\renewcommand{\shortauthors}{Pál András Papp and Roger Wattenhofer}

\begin{abstract}
We study financial networks where banks are connected by debt contracts. We consider the operation of \textit{debt swapping} when two creditor banks decide to exchange an incoming payment obligation, thus leading to a locally different network structure. We say that a swap is \textit{positive} if it is beneficial for both of the banks involved; we can interpret this notion either with respect to the amount of assets received by the banks, or their exposure to different shocks that might hit the system.
		
We analyze various properties of these swapping operations in financial networks. We first show that there can be no positive swap for any pair of banks in a static financial system, or when a shock hits each bank in the network proportionally. We then study worst-case shock models, when a shock of given size is distributed in the worst possible way for a specific bank. If the goal of banks is to minimize their losses in such a worst-case setting, then a positive swap can indeed exist. We analyze the effects of such a positive swap on other banks of the system, the computational complexity of finding a swap, and special cases where a swap can be found efficiently. Finally, we also present some results for more complex swapping operations when the banks swap multiple contracts, or when more than two banks participate in the swap.
\end{abstract}

\ccsdesc[500]{Applied computing~Economics}
\ccsdesc[100]{Theory of computation~Market equilibria}
\ccsdesc[100]{Theory of computation~Network games}

\maketitle

\section{Introduction} \label{sec:Intro} \genfootnote{The short version of the paper is published in the \textit{Proceedings of the 22nd ACM Conference on Economics and Computation (EC ’21), in 2021, Budapest, Hungary}, with ACM DOI 10.1145/3465456.3467638.}

Nowadays the world's financial system forms a highly interconnected network, where banks and other financial institutions are connected by various kinds of debt contracts. These interdependencies between the banks often introduce a systemic risk into the financial system, when e.g. the default of a single bank can cause a cascading effect through the network. These effects also played a major role in the financial crisis of 2008; as such, there has been a rapidly growing interest in the network-based analysis of these interbank systems over the last few years.

One of the most fundamental problems in these financial systems is the \textit{clearing problem}: given a specific amount of funds for each bank, and network of debt contracts between the banks, we need to decide how much of the these payment obligations the banks can fulfill. The solution of this problem, which is essentially a payment configuration over the network, is of crucial importance for the banks, since it specifies the amount of assets that they receive in the system.

Given such a financial network, it is natural to assume that banks would try to use different tools to influence the network in order to end up in a more beneficial situation. Since the payment rules are often fixed in a given network, one natural approach for that is to execute some modification on the network structure. The involved banks can have different possible motivations to support such a reorganization: it might directly increase the amount of assets they receive, or it might make their situation more resilient to an external shock that could hit the network.

A very natural candidate for such a network operation is \textit{debt swapping}. Given a debt contract from bank $u_1$ towards $v_1$, and another debt contract from $u_2$ to $v_2$, the creditor banks $v_1$ and $v_2$ might decide to swap their roles as the recipient of these contracts: $u_1$ will still owe the same amount of money as before, but now to $v_2$ instead of $v_1$, and $u_2$ will now owe the same amount of money to $v_1$ instead of $v_2$.

This swapping operation is a minor change in the network structure that only affects these two contracts. Since the debtors $u_1$ and $u_2$ have the same amount of liabilities as before, they do not have a direct reason to object to the operation. On the other hand, the two acting banks $v_1$ and $v_2$ only agree to execute the operation if it is mutually beneficial for them (according to some specific objective); in this case, we say that the swap is \textit{positive}.

Our main goal in the paper is to study the properties of debt swapping operations in financial networks. We first consider static financial systems without any kind of shock, analyzing whether two specific banks $v_1$ and $v_2$ can execute a swap that improves both of their situations. We show that such a positive swap is not possible in any network structure, i.e. any swap can only increase the assets of one of the acting banks.

We then consider different models of shock that might hit the financial network, and we analyze their effects from the perspective of a specific bank $v$. We consider $3$ different shock models in detail: (i) when each bank is hit proportionally by the shock, (ii) when $k$ specific banks are hit by the shock, but in the worst possible way for $v$, and (iii) when a shock of total size $\rho$ is distributed among the banks, again in the worst possible way for $v$. Our main goal is to investigate whether bank $v$ can reduce its exposure to such shocks with the debt swapping operation.

With respect to a proportional shock, we can again show that a positive swap is not possible in any network. However, in the other two (worst-case) shock models, a positive swap is indeed possible in some cases, so we study the properties of positive swaps in these models in more detail.

We first show that while a positive swap is beneficial for the two acting banks $v_1$ and $v_2$, it can result in a strictly worse situation for some third party banks in the network (with respect to shock exposure). Also, since computing the worst-case shock for a bank $v$ is in general an NP-hard problem, we study the special case of tree networks where the effects of a shock can still be computed efficiently. Finally, we show that the debt swapping operation is not sufficient to find every improvement opportunity in the network: there are cases where there is no positive swap for $v_1$ and $v_2$, but the banks could still improve their situation by executing a more sophisticated operation, e.g. swapping multiple contracts in one step, or also including a third bank in the debt reorganization.

\section{Related Work}

We study a popular model of financial systems with banks and debt contracts that was originally developed by Eisenberg and Noe \cite{model1}. It is known that in these networks, there exists a maximal clearing vector and it can be found in polynomial time \cite{model1, veraart}. Many works have later also studied the extensions of this base model by further aspects, such as default costs \cite{veraart}, cross-ownership relations \cite{cross1, cross2} or credit default swaps \cite{cds1, coveredCDS}.

The study of different properties of these network models has been rapidly gaining attention in the past decade. The most popular line of work focuses on the propagation of shocks through these networks, and whether larger connectivity amplifies or reduces these cascading effects \cite{cross2, gen1, gen2, gen3}. Others study how the clearing vector depends on a minor perturbation of a liability in the network \cite{sensitivity, worsttotal}. Some papers take a significantly different approach, such as analyzing the topic from a computational perspective \cite{cds2, arxiv} or as a dynamic process \cite{seqclear1, itcs}.

However, there are only a few works that focus on operations that banks could execute on the network in order to improve their situation. The most well-known approach to reduce systemic risks is the use of CCPs (central clearing counterparties), when a group of banks distribute risks among themselves by essentially introducing a new entity into the network \cite{CCP1, CCP2, CCP3, CCP4}. However, this assumes a major and centralized reorganization in the network, as opposed to our approach.

The closest line of work to our results is the analysis of portfolio compression by Schuldenzucker \textit{et al.} \cite{cycles}, also studied by \cite{cycles2, cycles3}; this is a technique where entire cycles of debts are removed from the network. The work of \cite{cycles, cycles2} extensively studies when such a compression operation is beneficial for the banks within the cycle, and for the remaining banks of the system. We later discuss the relationship of debt swapping to portfolio compression in more detail.

The work of Bertschinger \textit{et al.} \cite{gametheo} studies the motivation of banks in a different setting: instead of executing changes to the topology, banks are allowed to decide the order in which they fulfill their payment obligations. The authors of \cite{gametheo} conduct a thorough investigation of this scenario from a game-theoretic perspective. The work of \cite{icalp} takes a similar game-theoretic approach in a network model with more complex derivatives.

Our paper also has a connection to previous works that study different models of external shocks which might hit a financial network; most of these assume some stochastic shock distribution in the network \cite{gen1, gen2}. We mention the work of Hemenway \textit{et al.} in particular \cite{worsttotal}, which introduces the shock model that we refer to as worst-sum shock in our paper, and shows that finding the worst shock in this model is NP-hard.

We point out that the clearing problem also attracts significant attention in more practical projects, e.g. in the European Central Bank stress test framework \cite{ecb}.

\section{Model and motivation}

\subsection{Financial networks} \label{sec:def}

We consider the financial network model originally developed by Eisenberg and Noe, and studied thoroughly in the past two decades. In this model, the system consists of a set of \textit{banks} $B$, also referred to as nodes. We usually denote individual banks by $u$, $v$ or $w$, and the number of banks by $n:=|B|$. Each bank $v$ has a specific amount of \textit{funds} (sometimes also called external assets), denoted by $e_v$.

The banks in the system are connected by debt contracts. Each debt contract is between two specific banks $u$ and $v$, and obliges the debtor $u$ to pay a specific amount of money (the \textit{weight} or notional of the contract) to the creditor $v$. We use $l_{u,v}$ to denote this liability from $u$ to $v$, and understand it to be $0$ if there is no debt contract from $u$ to $v$. Similarly to previous work, we assume that $l_{u,u}=0$ for each $u \in B$, i.e. no bank enters into a contract with itself. The outgoing contracts for a bank $u$ define a total liability of $l_u = \sum_{v \in B} l_{u,v}$ for $u$.

However, the payment on a contract from $u$ to $v$ can be less then $l_{u,v}$ if $u$ is not able to fulfill all of its payment obligations, or in other words, if bank $u$ is \textit{in default}. In this case the \textit{recovery rate} of $u$ (denoted by $r_u$) is the portion of payment obligations that $u$ is still able to fulfill; hence $r_u<1$ exactly if $u$ is in default. When a bank $u$ is in default, the model assumes that $u$ must use all of its assets to make payments, and it must make these payments \textit{proportionally} to the payment obligations. Given a recovery rate of $r_u$ for $u$, this principle of proportionality implies that the actual payment $p_{u,v}$ on each outgoing contract will be $p_{u,v} = r_u \cdot l_{u,v}$.

Given the payments on each contract, the assets of a bank $u$ are defined as $a_u=e_u + \sum_{v \in B} p_{v,u}$. With bank $u$ having assets of $a_u$, according to our assumptions, the recovery rate $r_u$ must satisfy
\begin{displaymath}
 r_u=
\begin{cases}
    \: \, 1, & \quad \text{if } a_u \geq l_u \quad \text{  (i.e. if } u \text{ is not in default)} \\
		\: \, \frac{a_u}{l_u}, & \quad \text{if } a_u < l_u \quad \text{  (i.e. if } u \text{ is in default)}.
\end{cases}
\end{displaymath}
We say that a vector of recovery rates $r \in [0,1]^n$ is a \textit{clearing vector} (or \textit{equilibrium}) of the system if this property is fulfilled for each bank $u$, i.e. the recovery rate vector is consistent with the assets it defines in the network.

Previous work has shown that in any financial network, there is a clearing vector $r$ which maximizes the assets of each bank simultaneously, and this vector can be found in polynomial time \cite{veraart}. We will refer to this maximal clearing vector as the \textit{solution} of the system. Throughout the paper, we will understand $a_v$ to always refer to the assets of $v$ in this solution of the network.

For an example, consider the financial network shown in Figure \ref{fig:example}. In this system, bank $v_1$ has $e_{v_1}=4$ and a liability of $2$, so it will always be able to fulfill this payment obligation, regardless of the payment received from $v_4$. Bank $v_5$ is not in default either, since it has no liabilities at all. One the other hand, banks $v_3$ and $v_4$ are not able to fulfill their obligations even if they receive full payment on their incoming contracts, so they are certainly in default. This also means that the payment $p_{v_4, v_2}$ will be less than $2$, which also sends $v_2$ into default.

Recall that defaulting banks make payments proportionally to the liabilities; this implies that the assets of $v_2$, $v_3$, $v_4$ must satisfy $a_{v_2}=\frac{1}{2} \cdot a_{v_4}+2+1$, $a_{v_3}=a_{v_2}+2$ and $a_{v_4}=\frac{1}{3} \cdot a_{v_3}$. This leads to the payment configuration shown in the brackets in the right side of Figure \ref{fig:example}, and a solution with $a_{v_1}=5$, $a_{v_2}=4$, $a_{v_3}=6$, $a_{v_4}=2$ and $a_{v_5}=5$.

\begin{figure}
\centering
\minipage{0.45\textwidth}
\centering
\definecolor{dgray}{gray}{0.35}

\begin{tikzpicture}

	\draw[very thick, dgray, arrows=-latex] (0pt,0pt) -- (52pt,0pt);
	\draw[very thick, dgray, arrows=-latex] (60pt,0pt) -- (112pt,0pt);
	\draw[very thick, dgray, arrows=-latex] (120pt,0pt) -- (146pt,45pt);
	\draw[very thick, dgray, arrows=-latex] (120pt,0pt) -- (100pt,50pt) -- (83pt,50pt);
	\draw[very thick, dgray, arrows=-latex] (75pt,50pt) -- (64pt,5pt);
	\draw[very thick, dgray, arrows=-latex] (75pt,50pt) -- (30pt,50pt) -- (4pt,5pt);
	
	\draw[black, fill=white] (0pt,0pt) circle (8pt);
	\draw[black, fill=white] (60pt,0pt) circle (8pt);
	\draw[black, fill=white] (120pt,0pt) circle (8pt);
	\draw[black, fill=white] (75pt,50pt) circle (8pt);
	\draw[black, fill=white] (150pt,50pt) circle (8pt);
	
	\node[anchor=center] at (0.5pt,-0.5pt) {\normalsize $v_1$};
	\node[anchor=center] at (60.5pt,-0.5pt) {\normalsize $v_2$};
	\node[anchor=center] at (120.5pt,-0.5pt) {\normalsize $v_3$};
	\node[anchor=center] at (75.5pt,49.5pt) {\normalsize $v_4$};
	\node[anchor=center] at (150.5pt,49.5pt) {\normalsize $v_5$};
	
	\node[anchor=center] at (30pt,-8pt) {\small $2$};
	\node[anchor=center] at (90pt,-8pt) {\small $5$};
	\node[anchor=center] at (139pt,21pt) {\small $6$};
	\node[anchor=center] at (113pt,35pt) {\small $3$};
	\node[anchor=center] at (63pt,28pt) {\small $2$};
	\node[anchor=center] at (18pt,42pt) {\small $2$};

	
	\draw [fill=white] (3.5pt,-3.5pt) rectangle (9.5pt,-13.5pt);
	\node[anchor=center] at (6.5pt,-8pt) {\small $4$};
	\draw [fill=white] (63.5pt,-3.5pt) rectangle (69.5pt,-13.5pt);
	\node[anchor=center] at (66.5pt,-8pt) {\small $1$};
	\draw [fill=white] (123.5pt,-3.5pt) rectangle (129.5pt,-13.5pt);
	\node[anchor=center] at (126.5pt,-8pt) {\small $2$};
	\draw [fill=white] (78.5pt,46.5pt) rectangle (84.5pt,36.5pt);
	\node[anchor=center] at (81.5pt,42pt) {\small $0$};
	\draw [fill=white] (153.5pt,46.5pt) rectangle (159.5pt,36.5pt);
	\node[anchor=center] at (156.5pt,42pt) {\small $1$};

\end{tikzpicture}
\endminipage\hfill
\hspace{0.02\textwidth}
\minipage{0.52\textwidth}
\centering
\definecolor{dgray}{gray}{0.35}

\begin{tikzpicture}

	\draw[very thick, dgray, arrows=-latex] (0pt,0pt) -- (52pt,0pt);
	\draw[very thick, dgray, arrows=-latex] (60pt,0pt) -- (112pt,0pt);
	\draw[very thick, dgray, arrows=-latex] (120pt,0pt) -- (146pt,45pt);
	\draw[very thick, dgray, arrows=-latex] (120pt,0pt) -- (100pt,50pt) -- (83pt,50pt);
	\draw[very thick, dgray, arrows=-latex] (75pt,50pt) -- (64pt,5pt);
	\draw[very thick, dgray, arrows=-latex] (75pt,50pt) -- (30pt,50pt) -- (4pt,5pt);
	
	\draw[black, fill=white] (0pt,0pt) circle (8pt);
	\draw[black, fill=white] (60pt,0pt) circle (8pt);
	\draw[black, fill=white] (120pt,0pt) circle (8pt);
	\draw[black, fill=white] (75pt,50pt) circle (8pt);
	\draw[black, fill=white] (150pt,50pt) circle (8pt);
	
	\node[anchor=center] at (0.5pt,-0.5pt) {\normalsize $v_1$};
	\node[anchor=center] at (60.5pt,-0.5pt) {\normalsize $v_2$};
	\node[anchor=center] at (120.5pt,-0.5pt) {\normalsize $v_3$};
	\node[anchor=center] at (75.5pt,49.5pt) {\normalsize $v_4$};
	\node[anchor=center] at (150.5pt,49.5pt) {\normalsize $v_5$};
	

	\node[anchor=center] at (30pt,-8pt) {\scriptsize $[2]$};
	\node[anchor=center] at (90pt,-8pt) {\scriptsize $[4]$};
	\node[anchor=center] at (140pt,21pt) {\scriptsize $[4]$};
	\node[anchor=center] at (113.5pt,35pt) {\scriptsize $[2]$};
	\node[anchor=center] at (62.5pt,28pt) {\scriptsize $[1]$};
	\node[anchor=center] at (17pt,42pt) {\scriptsize $[1]$};
	
	\draw [fill=white] (3.5pt,-3.5pt) rectangle (9.5pt,-13.5pt);
	\node[anchor=center] at (6.5pt,-8pt) {\small $4$};
	\draw [fill=white] (63.5pt,-3.5pt) rectangle (69.5pt,-13.5pt);
	\node[anchor=center] at (66.5pt,-8pt) {\small $1$};
	\draw [fill=white] (123.5pt,-3.5pt) rectangle (129.5pt,-13.5pt);
	\node[anchor=center] at (126.5pt,-8pt) {\small $2$};
	\draw [fill=white] (78.5pt,46.5pt) rectangle (84.5pt,36.5pt);
	\node[anchor=center] at (81.5pt,42pt) {\small $0$};
	\draw [fill=white] (153.5pt,46.5pt) rectangle (159.5pt,36.5pt);
	\node[anchor=center] at (156.5pt,42pt) {\small $1$};

\end{tikzpicture}
\endminipage\hfill
\caption{Example financial network on $5$ banks, with the funds of banks shown in rectangles besides the banks. The left side shows the liabilities, while the right side shows the payments $p$ in the solution of the network (with the liabilities removed to avoid confusion).}
\label{fig:example}
\end{figure}
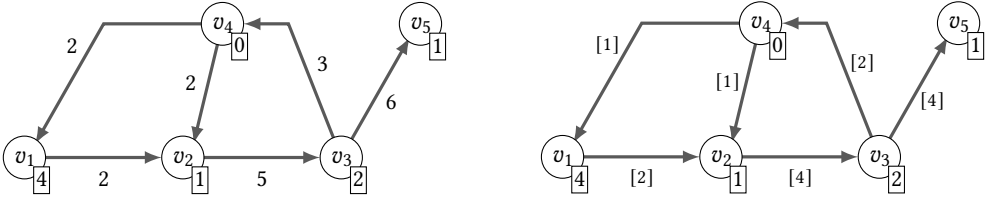

Throughout the paper, we use a simplified version of this notation to keep our figures easier to follow. We assume that all debt contracts have a weight of $1$ by default, and we only show the weight of a contract $(u,v)$ in our figures when $l_{u,v} \neq 1$. Similarly, whenever a bank $u$ has no funds at all ($e_u=0$), we do not explicitly show this in the figures. Occasionally, we also write $\infty$ to indicate a very large liability, but this can always be replaced by an appropriate constant value. Note that in order to have simple examples with unit-weight contracts, we will often work with fractional funds and assets; however, one can always scale up such networks by a constant factor to avoid this.

We assume that when executing operations on the network, the goal of a bank $v$ is to maximize $a_v$, i.e. its amount of assets in the solution of the resulting network. Since the total liabilities of a bank will never change in our setting, this is equivalent to maximizing the recovery rate (for defaulting banks) or the money remaining after payments (non-defaulting banks). The latter case is somewhat more motivated in practice, since it directly translates to more equity for bank $v$; we note that our results also apply in this more restricted setting. In particular, our non-existence results in Section \ref{sec:base} hold for either of the two cases, while most of our constructions in Section \ref{sec:worstcase} correspond to the second case of providing more equity for the acting banks.

Finally, we point out that in practice, the debt contracts in the network might be connected to earlier transactions between the banks, e.g. a debt from $u$ to $v$ might be due to a loan previously given by $v$ to $u$. However, the model assumes that in such cases, the amount received with this loan is already implicitly represented in the funds $e_u$, and thus the external assets and the liabilities are together sufficient to describe the current state of the system.

\subsection{Shock models}

While we also analyze how swapping can be used in a static network, our main goal is to study how this operation can mitigate the effects of an external shock that might hit the financial system. As such, we first define some simplified models of shock that we study in our networks.

In each of these shock models, we assume that some unforeseen event partially or completely removes the funds of some of the banks in the network; as such, the system will have a different solution after the shock where the assets of some banks are possibly lower than before. We then study the impact of a shock on a specific node $v$ as a \textit{shock function}, describing the assets $a_v$ remaining at $v$ in the solution of the network as a function of the size of the shock.

For the formal definitions of these shock models, let us use $a_v^{\,(G')}$ to denote to the assets of $v$ in some modified version $G'$ of our original financial network.

One natural model for market shocks is to assume that each bank is hit proportionally, i.e. each bank $u$ in the network loses a $\lambda$ portion of its original funds (or, equivalently, only retains a $(1-\lambda)$ portion of its original funds), for some parameter $\lambda$.

\begin{definition}
The \emph{proportional shock model} has a shock function $f_v : [0,1] \rightarrow \mathbb{R}^+_0$. For any $\lambda \in [0, 1]$, we consider the modified network $G_{\lambda}$ where each node $u \in B$ only has $e_u^{\, (G_{\lambda})} := (1-\lambda) \cdot e_u$, and we define $f_v(\lambda) = a_v^{\, (G_{\lambda})}$.
\end{definition}

Another approach is to study the system from a worst-case perspective. For instance, we can consider the cases when exactly $k$ specific banks are hit by the shock, but these banks lose all of their funds. Then from all the different possible $k$-tuples of banks that can be hit, we consider the combination which is the worst for bank $v$, i.e. the case where $v$ ends up with the smallest amount of assets.

\begin{definition}
The \emph{worst-set shock model} has a shock function $f_v : \{0, 1, ..., n\} \rightarrow \mathbb{R}^+_0$. For any integer $k \in [0, n]$, we consider the family $\mathcal{G}$ of all networks $G_k$ that can be obtained from $G$ by selecting a subset of $k$ banks $U$, and setting $e_u^{\, (G_k)} := 0$ if $u \in U$, and $e_u^{\, (G_k)} := e_u$ otherwise. We define
\begin{displaymath}
 f_v(k) = \min_{\, G_k \in \mathcal{G}} \: a_v^{\, (G_k)} \, .
\end{displaymath}
\end{definition}

Finally, another model for worst-case analysis is the worst-sum shock model, which has already been studied before in \cite{worsttotal}. In this case, we consider all shocks of total size $\rho$, i.e. where banks of the system lose a total of $\rho$ funds altogether, and we assume that these $\rho$ losses are distributed among the banks in the worst possible way in terms of the outcome for $v$.

\begin{definition}
The \emph{worst-sum shock model} assumes a shock function $f_v : [0, \sum_{u \in G} e_u] \rightarrow \mathbb{R}^+_0$ and a parameter $\rho \in [0, \sum_{u \in G} e_u]$. We consider the family $\mathcal{G}$ of all networks $G_{\rho}$ that satisfy
\begin{displaymath}
 \sum_{u \in G} e_u - \sum_{u \in G_{\rho}} \: e_u^{\, (G_{\rho})} = \rho \, ,
\end{displaymath}
while also having $0 \leq e_u^{\: (G_{\rho})} \leq e_u$ for each bank $u \in G_{\rho}$. We define
\begin{displaymath}
 f_v(\rho) = \min_{G_{\rho} \in \mathcal{G}} \: a_v^{\, (G_{\rho})} \, .
\end{displaymath}
\end{definition}

Naturally, all these shock functions are monotonically decreasing. As such, for simplicity, we will often only discuss or illustrate the shock functions until the point where they first decrease to $0$, since the function value always remains $0$ from this point.

It is also natural to consider more realistic variants of these models by assuming that the shock will be of limited size. For example, we can introduce a \textit{limited worst-set} model where there is an upper limit $K$ on the number of banks that are hit by a shock, for some integer parameter $K \leq n$. This model comes with the same shock function, but restricted to the domain $\{0, 1, ..., K\}$.

\subsection{Motivation for swapping} \label{sec:motive}

Introducing these shock models already allows us to demonstrate the motivation behind debt swapping with a simple example. Consider the financial network in Figure \ref{fig:motive}, where the assets of banks $v_1$ and $v_2$ are indirectly dependent on the well-being of banks $s_1$ and $s_2$, respectively.

In this initial state of the system, the shock functions of both $v_1$ and $v_2$ are as shown in the upper row diagrams. For example, in the proportional shock model, $v_1$ only starts losing assets when $s_1$ is unable to pay all of its liabilities, which happens at $\lambda>\frac{1}{2}$ (since this implies $(1-\lambda) \cdot e_{s_1}<2$); after this point, $v_1$ keeps continuously losing assets until $\lambda=1$. In the worst-set model, $v_1$ already loses all of its assets for $k=1$, since the worst-case for $v_1$ is when $s_1$ is hit by the shock (and set to $e_{s_1}=0$). In the worst-sum model, the first $2$ units of loss at $s_1$ do not affect $v_1$, since $s_1$ can still fulfill its obligations; from this point, $a_{v_1}$ decreases linearly until $\rho=4$.

Now assume that banks $v_1$ and $v_2$ decide to swap one of their incoming debt contracts. That is, they agree that the payment obligation of $1$ from $u_1$ should now go towards banks $v_2$ instead of $v_1$, and in return, the payment obligation of $1$ from $u_2$ should now go towards bank $v_1$ instead of $v_2$. Note that at first glance, this does not affect the situation of any bank in the system: each bank will have the same amount of liabilities and the same total of incoming payment obligations as initially. In this particular case, the assets of $v_1$ and $v_2$ do not change either: we will also have $a_{v_1}'=a_{v_2}'=2$ in the solution of the new network obtained after the operation.

However, the swap does improve the situation of the nodes with respect to the different shock models, as shown in the lower row diagrams. Intuitively, both $v_1$ and $v_2$ has now diversified their dependencies on the rest of the network, so a shock at a specific point in the network does not affect them as drastically as before.

\begin{figure}
\centering
\minipage{0.45\textwidth}
\centering
	\resizebox{1.0\textwidth}{!}{\definecolor{dgray}{gray}{0.35}

\begin{tikzpicture}

	\draw[very thick, dgray, arrows=-latex] (0pt,15pt) -- (28pt,2pt);
	\draw[very thick, dgray, arrows=-latex] (0pt,15pt) -- (28pt,28pt);
	\draw[very thick, dgray, arrows=-latex] (0pt,70pt) -- (28pt,57pt);
	\draw[very thick, dgray, arrows=-latex] (0pt,70pt) -- (28pt,83pt);
	\draw[very thick, dgray, arrows=-latex] (35pt,0pt) -- (73pt,18pt);
	\draw[very thick, dgray, arrows=-latex] (35pt,30pt) -- (73pt,26pt);
	\draw[very thick, dgray, arrows=-latex] (35pt,55pt) -- (73pt,59pt);
	\draw[very thick, dgray, arrows=-latex] (35pt,85pt) -- (73pt,67pt);

	\draw[black, fill=white] (0pt,15pt) circle (8pt);
	\draw[black, fill=white] (0pt,70pt) circle (8pt);
	\draw[black, fill=white] (35pt,0pt) circle (8pt);
	\draw[black, fill=white] (35pt,30pt) circle (8pt);
	\draw[black, fill=white] (35pt,55pt) circle (8pt);
	\draw[black, fill=white] (35pt,85pt) circle (8pt);
	\draw[black, fill=white] (80pt,22pt) circle (8pt);
	\draw[black, fill=white] (80pt,63pt) circle (8pt);
	
	\node[anchor=center] at (0.5pt,14.5pt) {\normalsize $s_2$};
	\node[anchor=center] at (0.5pt,69.5pt) {\normalsize $s_1$};
	\node[anchor=center] at (35.5pt,-0.5pt) {\normalsize $u_3$};
	\node[anchor=center] at (35.5pt,29.5pt) {\normalsize $u_2$};
	\node[anchor=center] at (35.5pt,54.5pt) {\normalsize $u_1$};
	\node[anchor=center] at (35.5pt,84.5pt) {\normalsize $u_0$};
	\node[anchor=center] at (80.5pt,21.5pt) {\normalsize $v_2$};
	\node[anchor=center] at (80.5pt,62.5pt) {\normalsize $v_1$};
	
	\draw [fill=white] (3.5pt,11.5pt) rectangle (9.5pt,1.5pt);
	\node[anchor=center] at (6.5pt,7pt) {\normalsize $4$};
	\draw [fill=white] (3.5pt,66.5pt) rectangle (9.5pt,56.5pt);
	\node[anchor=center] at (6.5pt,61.5pt) {\normalsize $4$};
	
	
	\draw[line width=3pt, arrows=-stealth] (105pt,42.5pt) -- (130pt,42.5pt);
	
	
	\draw[very thick, dgray, arrows=-latex] (150pt,15pt) -- (178pt,2pt);
	\draw[very thick, dgray, arrows=-latex] (150pt,15pt) -- (178pt,28pt);
	\draw[very thick, dgray, arrows=-latex] (150pt,70pt) -- (178pt,57pt);
	\draw[very thick, dgray, arrows=-latex] (150pt,70pt) -- (178pt,83pt);
	\draw[very thick, dgray, arrows=-latex] (185pt,0pt) -- (223pt,18pt);
	\draw[very thick, dgray, arrows=-latex] (185pt,30pt) -- (223pt,59pt);
	\draw[very thick, dgray, arrows=-latex] (185pt,55pt) -- (223pt,26pt);
	\draw[very thick, dgray, arrows=-latex] (185pt,85pt) -- (223pt,67pt);

	\draw[black, fill=white] (150pt,15pt) circle (8pt);
	\draw[black, fill=white] (150pt,70pt) circle (8pt);
	\draw[black, fill=white] (185pt,0pt) circle (8pt);
	\draw[black, fill=white] (185pt,30pt) circle (8pt);
	\draw[black, fill=white] (185pt,55pt) circle (8pt);
	\draw[black, fill=white] (185pt,85pt) circle (8pt);
	\draw[black, fill=white] (230pt,22pt) circle (8pt);
	\draw[black, fill=white] (230pt,63pt) circle (8pt);
	
	\node[anchor=center] at (150.5pt,14.5pt) {\normalsize $s_2$};
	\node[anchor=center] at (150.5pt,69.5pt) {\normalsize $s_1$};
	\node[anchor=center] at (185.5pt,-0.5pt) {\normalsize $u_3$};
	\node[anchor=center] at (185.5pt,29.5pt) {\normalsize $u_2$};
	\node[anchor=center] at (185.5pt,54.5pt) {\normalsize $u_1$};
	\node[anchor=center] at (185.5pt,84.5pt) {\normalsize $u_0$};
	\node[anchor=center] at (230.5pt,21.5pt) {\normalsize $v_2$};
	\node[anchor=center] at (230.5pt,62.5pt) {\normalsize $v_1$};
	
	\draw [fill=white] (153.5pt,11.5pt) rectangle (159.5pt,1.5pt);
	\node[anchor=center] at (156.5pt,7pt) {\normalsize $4$};
	\draw [fill=white] (153.5pt,66.5pt) rectangle (159.5pt,56.5pt);
	\node[anchor=center] at (156.5pt,61.5pt) {\normalsize $4$};
	
\end{tikzpicture}}
\endminipage\hfill
\hspace{0.02\textwidth}
\minipage{0.52\textwidth}
\centering
	\resizebox{1.0\textwidth}{!}{\definecolor{dgray}{gray}{0.35}

\begin{tikzpicture}

	\draw[thick, arrows=-latex] (0pt,-15pt) -- (0pt,40pt);
	\draw[thick, arrows=-latex] (-5pt,-10pt) -- (70pt,-10pt);
	
	\draw[very thick, blue] (0pt,30pt) -- (30pt,30pt) -- (60pt,-10pt);
	
	\node[anchor=center] at (-4pt,-15pt) {\small $0$};
	
	\draw[thick] (-2pt,10pt) -- (2pt,10pt);
	\draw[thick] (-2pt,30pt) -- (2pt,30pt);
	\node[anchor=center] at (-5pt,10pt) {\small $1$};
	\node[anchor=center] at (-5pt,30pt) {\small $2$};
	
	\draw[thick] (30pt,-12pt) -- (30pt,-8pt);
	\draw[thick] (60pt,-12pt) -- (60pt,-8pt);
	\node[anchor=center] at (30pt,-17pt) {\small $0.5$};
	\node[anchor=center] at (60pt,-17pt) {\small $1$};
	
	\node[anchor=center] at (30pt,-32pt) {\normalsize \textit{proportional}};
	
	
	\draw[thick, arrows=-latex] (100pt,-15pt) -- (100pt,40pt);
	\draw[thick, arrows=-latex] (95pt,-10pt) -- (170pt,-10pt);
	
	\node[anchor=center] at (96pt,-15pt) {\small $0$};
	
	\draw[thick] (98pt,10pt) -- (102pt,10pt);
	\draw[thick] (98pt,30pt) -- (102pt,30pt);
	\node[anchor=center] at (95pt,10pt) {\small $1$};
	\node[anchor=center] at (95pt,30pt) {\small $2$};
	
	\draw[thick] (130pt,-12pt) -- (130pt,-8pt);
	\draw[thick] (160pt,-12pt) -- (160pt,-8pt);
	\node[anchor=center] at (130pt,-17pt) {\small $1$};
	\node[anchor=center] at (160pt,-17pt) {\small $2$};
	
	\draw[blue, fill=blue] (100pt,30pt) circle (2pt);
	\draw[blue, fill=blue] (130pt,10pt) circle (2pt);
	\draw[blue, fill=blue] (160pt,-10pt) circle (2pt);
	
	\node[anchor=center] at (130pt,-32pt) {\normalsize \textit{worst-set}};
	
	
	\draw[thick, arrows=-latex] (200pt,-15pt) -- (200pt,40pt);
	\draw[thick, arrows=-latex] (195pt,-10pt) -- (290pt,-10pt);
	
	\draw[very thick, blue] (200pt,30pt) -- (220pt,30pt) -- (240pt,10pt) -- (260pt,10pt) -- (280pt,-10pt);
	
	\node[anchor=center] at (196pt,-15pt) {\small $0$};
	
	\draw[thick] (198pt,10pt) -- (202pt,10pt);
	\draw[thick] (198pt,30pt) -- (202pt,30pt);
	\node[anchor=center] at (195pt,10pt) {\small $1$};
	\node[anchor=center] at (195pt,30pt) {\small $2$};
	
	\draw[thick] (220pt,-12pt) -- (220pt,-8pt);
	\draw[thick] (240pt,-12pt) -- (240pt,-8pt);
	\draw[thick] (260pt,-12pt) -- (260pt,-8pt);
	\draw[thick] (280pt,-12pt) -- (280pt,-8pt);
	\node[anchor=center] at (220pt,-17pt) {\small $2$};
	\node[anchor=center] at (240pt,-17pt) {\small $4$};
	\node[anchor=center] at (260pt,-17pt) {\small $6$};
	\node[anchor=center] at (280pt,-17pt) {\small $8$};
	
	\node[anchor=center] at (240pt,-32pt) {\normalsize \textit{worst-sum}};
	
	
	\draw[thick, arrows=-latex] (0pt,75pt) -- (0pt,130pt);
	\draw[thick, arrows=-latex] (-5pt,80pt) -- (70pt,80pt);
	
	\draw[very thick, blue] (0pt,120pt) -- (30pt,120pt) -- (60pt,80pt);
	
	\node[anchor=center] at (-4pt,75pt) {\small $0$};
	
	\draw[thick] (-2pt,100pt) -- (2pt,100pt);
	\draw[thick] (-2pt,120pt) -- (2pt,120pt);
	\node[anchor=center] at (-5pt,100pt) {\small $1$};
	\node[anchor=center] at (-5pt,120pt) {\small $2$};
	
	\draw[thick] (30pt,78pt) -- (30pt,82pt);
	\draw[thick] (60pt,78pt) -- (60pt,82pt);
	\node[anchor=center] at (30pt,73pt) {\small $0.5$};
	\node[anchor=center] at (60pt,73pt) {\small $1$};
	
	\node[anchor=center] at (30pt,58pt) {\normalsize \textit{proportional}};
	
	
	\draw[thick, arrows=-latex] (100pt,75pt) -- (100pt,130pt);
	\draw[thick, arrows=-latex] (95pt,80pt) -- (170pt,80pt);
	
	\node[anchor=center] at (96pt,75pt) {\small $0$};
	
	\draw[thick] (98pt,100pt) -- (102pt,100pt);
	\draw[thick] (98pt,120pt) -- (102pt,120pt);
	\node[anchor=center] at (95pt,100pt) {\small $1$};
	\node[anchor=center] at (95pt,120pt) {\small $2$};
	
	\draw[thick] (130pt,78pt) -- (130pt,82pt);
	\draw[thick] (160pt,78pt) -- (160pt,82pt);
	\node[anchor=center] at (130pt,73pt) {\small $1$};
	\node[anchor=center] at (160pt,73pt) {\small $2$};
	
	\draw[blue, fill=blue] (100pt,120pt) circle (2pt);
	\draw[blue, fill=blue] (130pt,80pt) circle (2pt);
	\draw[blue, fill=blue] (160pt,80pt) circle (2pt);
	
	\node[anchor=center] at (130pt,58pt) {\normalsize \textit{worst-set}};
	
	
	\draw[thick, arrows=-latex] (200pt,75pt) -- (200pt,130pt);
	\draw[thick, arrows=-latex] (195pt,80pt) -- (290pt,80pt);
	
	\draw[very thick, blue] (200pt,120pt) -- (220pt,120pt) -- (240pt,80pt);
	
	\node[anchor=center] at (196pt,75pt) {\small $0$};
	
	\draw[thick] (198pt,100pt) -- (202pt,100pt);
	\draw[thick] (198pt,120pt) -- (202pt,120pt);
	\node[anchor=center] at (195pt,100pt) {\small $1$};
	\node[anchor=center] at (195pt,120pt) {\small $2$};
	
	\draw[thick] (220pt,78pt) -- (220pt,82pt);
	\draw[thick] (240pt,78pt) -- (240pt,82pt);
	\draw[thick] (260pt,78pt) -- (260pt,82pt);
	\draw[thick] (280pt,78pt) -- (280pt,82pt);
	\node[anchor=center] at (220pt,73pt) {\small $2$};
	\node[anchor=center] at (240pt,73pt) {\small $4$};
	\node[anchor=center] at (260pt,73pt) {\small $6$};
	\node[anchor=center] at (280pt,73pt) {\small $8$};
	
	\node[anchor=center] at (240pt,58pt) {\normalsize \textit{worst-sum}};
	
\end{tikzpicture}}
\endminipage\hfill
\caption{Example where the acting nodes $v_1$ and $v_2$ swap one of their incoming debt contracts. The diagrams show the shock functions of the acting nodes before swapping (top row) and after swapping (bottom row).}
\label{fig:motive}
\end{figure}
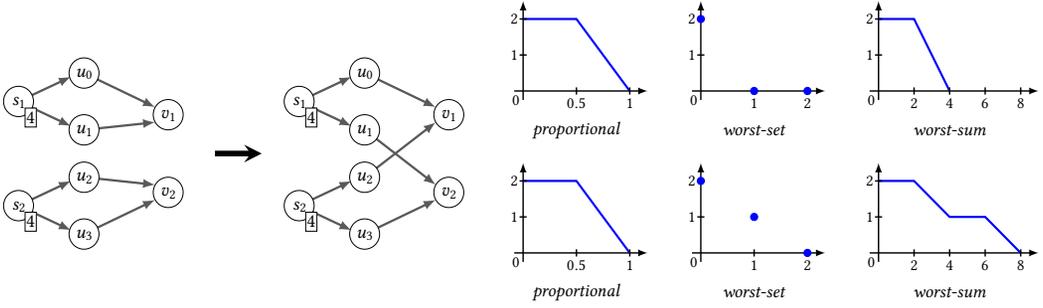

More specifically, in case of proportional shocks, the shock function of swapping nodes does not improve: they still start losing assets at the same pace from $\lambda=\frac{1}{2}$, since the proportional shock model always hits both $s_1$ and $s_2$ to the same extent. However, in the worst-set model, a shock of $k=1$ can now only affect one of the source nodes, so even with such a shock, both $v_1$ and $v_2$ will receive half of their assets. Similarly, in the worst-sum model, a shock of $\rho=4$ can now only remove the assets of either $s_1$ or $s_2$ in the worst case, but after this, the assets of $v_1$ and $v_2$ will again not decrease further until $\rho=6$, when the shock becomes large enough and the other source bank also becomes unable to meet its obligations. As such, in these latter two cases, the operation results in a strict improvement in the shock function of the acting nodes.

This example demonstrates that even in a relatively simple network, a swapping operation can easily ensure that banks are less exposed to an unforeseen event in the financial system. However, in the general case of more complex and interconnected networks, it is more difficult to identify such opportunities that mutually improve the banks' situation, or even to evaluate how volatile a specific network configuration is under given shock models.

\paragraph{Comparison to portfolio compression.}
We point out that in some sense, debt swapping is a similar operation portfolio compression, which was studied in detail in the works of \cite{cycles, cycles2}. In portfolio compression, the main idea is to locate debt cycles in the network topology, and remove these cycles to reduce the total amount of debt in the network; as such, this is also an operation that executes a minor reorganization step in the network structure. 

However, we argue that debt swapping is a significantly more justifiable operation from a fairness perspective. It is known that portfolio compression can yield a worse outcome for banks that are not contained in the cycle (the work of \cite{cycles} discusses this in a more general network model with default costs, but one can create similar examples in our model). Intuitively, the removal of a liability cycle from the network essentially means that the participating nodes pay these cyclic debts ``in advance'', i.e. earlier than their remaining liabilities, which implicitly contradicts the principle of proportionality; as such, a creditor of these nodes who is outside the cycle may end up receiving less money when a cycle node goes into default.

In contrast to this, debt swapping ensures that the total amount of incoming and outgoing liabilities remains the same at every bank in the network. The debtor banks of the swapped contracts still have the same amount of liabilities on these contracts, only now towards a different recipient. The acting banks (the creditors of the swapped contracts) only execute the operation if they both agree to it voluntarily. As such, no bank has a direct reason to object to this operation; in this situation, any (positive or possibly negative) effect on banks is really due to more subtle reasons originating from the change of the network structure. Hence we can argue that if many legal frameworks even allow the more questionable practice of portfolio compression, then debt swapping for risk mitigation should also be a reasonable and permitted tool for banks in practice.

On the other hand, portfolio compression also has a slight advantage over debt swapping from a regulatory perspective: it implicitly ensures that the debtor banks also approve the operation, since the banks in the cycle are both the creditors and the debtors of the removed contracts. In contrast to this, debt swapping assumes that the approval of the creditor banks is already sufficient to execute the operation.

\subsection{Swapping and terminology}

Finally, let us formally define debt swapping, as well as some other closely related operations.

\begin{definition}[Swapping]
Assume we have a network $G$ and four distinct nodes $v_1$, $v_2$, $u_1$ and $u_2$ that satisfy $l_{u_1, v_1}=l_{u_2, v_2}=d$ for some value $d$, and also $l_{u_1, v_2}=l_{u_2, v_1}=0$.

The \emph{swapping} of contracts $(u_1, v_1)$ and $(u_2, v_2)$ produces a new network $G'$, where the funds and liabilities are the same as in $G$, with the following exceptions: $l_{u_1, v_2}'=l_{u_2, v_1}'=d$, and $l_{u_1, v_1}'=l_{u_2, v_2}'=0$.
\end{definition}

We will refer to the nodes $v_1$ and $v_2$ as \textit{acting nodes} or swapping nodes. In general, we use the prime symbol (e.g. $a'_u$, $e_u'$ or $l_{u,v}'$) to refer to the properties of the network obtained after swapping.

We note that it would also be possible to extend this definition to allow swapping a pair of contracts with different weights; such a generalization would not affect most of our results. The acting nodes may still be willing to execute such a swap, since it can still improve their situation if the debtors are known to be in default, e.g. if we modify Figure \ref{fig:motive} such that the swapped contracts have arbitrary different weights that are both larger than 1. However, such an operation might seem less desirable from a regulator's perspective, since the total amount of incoming obligations changes for the banks. As such, we focus on the more convenient case where we only swap contracts of the same weight.

We also need to define whether we consider a specific operation beneficial, i.e. the condition when banks are willing to agree to a specific operation. However, this is not necessarily straightforward. Since all of our shocks are functions of a parameter, it can naturally happen that e.g. a swap of node $v_1$ provides $a_{v_1}'>a_{v_1}$ for a proportional shock of $\lambda_1=0.2$, but $a_{v_1}'<a_{v_1}$ for a shock of $\lambda_2=0.5$. In practice, $v_1$ might still agree to such an operation, e.g. if it is expecting that a shock of $\lambda_1=0.2$ will soon hit the market.

However, we assume that this is not the case, i.e. that banks do not have any assumptions on the values of the shock parameters. That is, we take a stricter stance, and assume that banks only consider an operation beneficial if their situation improves (or remains the same) for any possible value of the shock parameter.

More formally, given functions $f, g$ over the same domain $D$, let us say that $f \geq g$ if $\forall x \in D $ we have $ f(x) \geq g(x)$, and let us say that $f > g$ if $f \geq g$ and $\exists x \in D$ such that $f(x) > g(x)$. In this terminology, we assume that bank $v$ is only in favor of an operation if this provides a new shock function $f_v'$ such that $f_v' > f_v$: in this case, it is clear that the new situation is strictly more favorable to $v$.

\begin{definition}[Positive Swap]
We say that a swap of banks $v_1$ and $v_2$ is \emph{positive} according to a shock function $f$ if $f_{v_1}'>f_{v_1}$ and $f_{v_2}'>f_{v_2}$. We say that a swap is \emph{semi-positive} if $f_{v_1}' \geq f_{v_1}$ and $f_{v_2}' \geq f_{v_2}$, and at least one of the inequalities is strict.
\end{definition}

Finally, we can also consider some generalizations of swapping. One natural candidate for such an operation is when the banks not only exchange a pair of incoming contracts, but two sets of incoming contracts; we refer to this as \textit{portfolio swapping}. That is, given a set of banks $U_1$ who have debts towards $v_1$ (with a total weight of $d$), and a set of banks $U_2$ who have debts towards $v_2$ (also of total weight $d$), the operation creates a new network $G'$ where the creditor of all these contracts from $U_1$ (and $U_2$) becomes $v_2$ (and $v_1$, respectively). We can then define positivity (or semi-positivity) for such an operation in an analogous way.

Another possible generalization is to consider a multi-party swap operation where more than two banks participate; we will refer to this as \textit{debt reorganization}. That is, given a set of contracts $(u_i, v_i)$ of the same weight $l_{u_i, v_i}=d$, the acting nodes $v_i$ switch to a different permutation of the recipients of these contracts, e.g. changing the original liabilities $l_{u_1,v_1}=l_{u_2,v_2}=l_{u_3,v_3}=d$ into $l'_{u_1,v_2}=l'_{u_2,v_3}=l'_{u_3,v_1}=d$. Once again, a debt reorganization is positive if we have $f_{v_i}'>f_{v_i}$ for all the acting banks $v_i$, and semi-positive if $\forall i \; f_{v_i}' \geq f_{v_i}$ and $\exists i \; f_{v_i}' > f_{v_i}$.

For completeness, we also provide a formal definition of these more general operations in Appendix \ref{App:A}.

Finally, let us note that a natural generalization of this setting would be to also allow banks to swap only \textit{a portion} of their debts; more formally, we could select a value $d$ such that $d \leq l_{u_1, v_1}$, $d \leq l_{u_2, v_2}$, and then define the network $G'$ after swapping as $l_{u_1, v_1}'=l_{u_1, v_1\!}-d$, $l_{u_2, v_2}'=l_{u_2, v_2\!}-d$ and $l_{u_1,v_2}=l_{u_2,v_1}=d$. This provides significantly more opportunities for finding a positive swap in our networks; on the other hand, it requires us to split some payment contracts, thus possibly raising other problems in practice. We point out that most of our results can also be extended to this more general setting, sometimes with the extra technical step of introducing auxiliary nodes. However, for simplicity, we only focus on swapping entire debts in the paper.

\section{Swapping without shocks} \label{sec:base}

We begin our analysis by discussing debt swapping in financial networks without any kind of shock; that is, we consider a static financial system as defined in Section \ref{sec:def}, and we investigate whether banks can use swapping to ensure that they receive more assets in the modified network. This not only gives valuable insights into the properties of swapping in general, but it will also have direct implications on swapping in the proportional shock model.

Formally, we can describe this \textit{base model} as trivial shock function $f_v : \{ 0 \} \rightarrow \mathbb{R}^+_0$ that assigns $f_v(0)=a_v$ to the single point of its domain; in this model, a swap is simply positive if we have $a'(v_1)>a(v_1)$ and $a'(v_2)>a(v_2)$.

\subsection{Properties of financial systems}

We first establish some basic properties of financial systems that will be essential tools in the proofs that follow. Most of these properties have been mentioned or discussed to some extent in previous works, either in our model or an extended model with default costs; however, we formalize these properties for completeness, and we discuss them in more detail in Appendix \ref{App:A}.

In all of these properties, we consider a source node $s$ and a sink node $t$. We increase the funds $e_s$ of $s$ to some new value $\hat{e}_s$, and we study how this affects $t$, i.e. how the new assets $\hat{a}_t$ of $t$ relate to its original assets $a_t$. Note that we use this new notation $\hat{a}_t$ to make a clear distinction from $a_t'$, which will always refer to a state after swapping.

\begin{lemma}[Monotonicity] \label{lem:monot}
Assume we increase the funds of $s$ to $\hat{e}_s > e_s$. Then we have $\hat{a}_t \geq a_t$.
\end{lemma}

This property can also be extended to non-source and non-sink nodes, as well as the case when we increase the funds of multiple banks.

\begin{lemma}[Non-expansivity] \label{lem:nonexp}
Assume we set $\hat{e}_s=e_s + \Delta$ for some $\Delta>0$. Then $\hat{a}_t \leq a_t + \Delta$.
\end{lemma}

Note that for this property to hold, it is crucial that $t$ is a sink of the network. For example, consider the system of Figure \ref{fig:expansive}. Given a parameter $x \in [0, \frac{1}{2}]$ for the funds of $s$, the solution of the system is to have a payment of $2x$ from $u_1$ to $u_2$, and a payment of $x$ on every other contract. This means that an increase of $\hat{e}_s=e_s + \Delta$ results in $\hat{a}_{u_1} = a_{u_1} + 2 \cdot \Delta$.

We also introduce a specific notion for the case when the increase is as high as possible.

\begin{definition} \label{def:lin}
If setting $\hat{e}_s=e_s + \Delta$ gives $\hat{a}_t = a_t + \Delta$, then we say that $s$ is \emph{$t$-linear} on $[e_s, e_s+\Delta]$.
\end{definition}

\begin{figure}
\centering
\hspace{0.05\textwidth}
\minipage{0.4\textwidth}
\centering
	\vspace{33pt}
	\definecolor{dgray}{gray}{0.35}

\begin{tikzpicture}

	\draw[very thick, dgray, arrows=-latex] (0pt,0pt) -- (32pt,0pt);
	\draw[very thick, dgray, arrows=-latex] (40pt,0pt) -- (72pt,0pt);
	\draw[very thick, dgray, arrows=-latex] (80pt,0pt) -- (112pt,0pt);
	\draw[very thick, dgray, arrows=-latex] (80pt,0pt) -- (64pt,35pt);
	\draw[very thick, dgray, arrows=-latex] (60pt,40pt) -- (44pt,5pt);

	\draw[black, fill=white] (0pt,0pt) circle (8pt);
	\draw[black, fill=white] (40pt,0pt) circle (8pt);
	\draw[black, fill=white] (80pt,0pt) circle (8pt);
	\draw[black, fill=white] (120pt,0pt) circle (8pt);
	\draw[black, fill=white] (60pt,40pt) circle (8pt);
	
	\node[anchor=center] at (0pt,0pt) {\normalsize $s$};
	\node[anchor=center] at (40.5pt,-0.5pt) {\normalsize $u_1$};
	\node[anchor=center] at (80.5pt,-0.5pt) {\normalsize $u_2$};
	\node[anchor=center] at (120pt,0pt) {\normalsize $t$};
	
	\draw [fill=white] (3pt,-3.5pt) rectangle (9.5pt,-12pt);
	\node[anchor=center] at (6.25pt,-7.5pt) {\small $x$};
	
	
\end{tikzpicture}
	\caption{Counterexample for non-expansivity on a non-sink node of the network.}
	\label{fig:expansive}
\endminipage\hfill
\hspace{0.13\textwidth}
\minipage{0.35\textwidth}
\centering
	\vspace{33pt}
	\definecolor{dgray}{gray}{0.35}

\begin{tikzpicture}

	\draw[very thick, dgray, arrows=-latex] (0pt,0pt) -- (62pt,0pt);
	\draw[very thick, dgray, arrows=-latex] (0pt,40pt) -- (0pt,8pt);
	\draw[very thick, dgray, arrows=-latex] (70pt,40pt) -- (70pt,8pt);

	\draw[black, fill=white] (0pt,0pt) circle (8pt);
	\draw[black, fill=white] (0pt,40pt) circle (8pt);
	\draw[black, fill=white] (70pt,0pt) circle (8pt);
	\draw[black, fill=white] (70pt,40pt) circle (8pt);
	
	\node[anchor=center] at (0.5pt,-0.5pt) {\normalsize $v_1$};
	\node[anchor=center] at (0.5pt,39.5pt) {\normalsize $u_1$};
	\node[anchor=center] at (70.5pt,-0.5pt) {\normalsize $v_2$};
	\node[anchor=center] at (70.5pt,39.5pt) {\normalsize $u_2$};
	
	
	\draw [fill=white] (3.5pt,36pt) rectangle (9.5pt,24pt);
	\node[anchor=center] at (6.5pt,30pt) {\footnotesize $\frac{1}{2}$};
	\draw [fill=white] (73.5pt,35.5pt) rectangle (79.5pt,25.5pt);
	\node[anchor=center] at (76.5pt,31pt) {\small $1$};
	
\end{tikzpicture}
	\vspace{5pt}
	\caption{Example for a semi-positive swap in the base model.}
	\label{fig:semipos}
\endminipage\hfill
\hspace{0.05\textwidth}
\end{figure}

Linearity is a very useful property because it means, intuitively speaking, that all the extra funds given to $s$ will end up at $t$ after traveling a shorter or longer route through the network. Similarly, we can also define linearity on multiple target nodes: e.g. we say that $s$ is $(t_1, t_2)$-linear if when setting $\hat{e}_s=e_s + \Delta$, we have $(\hat{a}_{t_1} + \hat{a}_{t_2}) - (a_{t_1} + a_{t_2}) = \Delta$. In this case, there exist coefficients $\alpha_1, \alpha_2 \in (0,1)$ with $\alpha_1 + \alpha_2 = 1$ such that for any $\delta \in [0, \Delta]$, setting $\hat{e}_s=e_s + \delta$ results in $\hat{a}_{t_1} = a_{t_1} + \alpha_1 \cdot \delta$ and $\hat{a}_{t_2} = a_{t_2} + \alpha_2 \cdot \delta$.

\begin{lemma}[Concavity] \label{lem:conc}
If $s$ is $t$-linear on $[x, x+\Delta]$ for some $x$, then $s$ is also $t$-linear on $[0, x+\Delta]$.
\end{lemma}

Intuitively, this is because linearity can only hold if for any $e_s<x+\Delta$, there are still unpaid liabilities on every edge that is contained in some directed path from $s$ to $t$. However, in this case, the first $x$ funds of $s$ are also distributed along these paths according to the same proportions; this also means that the first $x$ funds of $s$ also arrive at $t$.

\subsection{First observations on swapping}

The example system of Figure \ref{fig:motive} gives the impression that when swapping in the base model, the acting nodes simply exchange a fixed amount of payments that are incoming on the swapped contracts; this might suggest that a positive (or semi-positive) swap cannot exist at all. This is indeed true when there is no directed path between $v_1$ and $v_2$ in either direction. However, once there is a directed path of debts between the banks, then semi-positive swaps are already possible, even in the simple case when the network is a directed acyclic graph (DAG).

\begin{lemma}
In a DAG network, there can be a semi-positive swap in the base model.
\end{lemma}

\begin{proof}
Consider the system in Figure \ref{fig:semipos}. In the initial state, the acting nodes have $a_{v_1}=\frac{1}{2}$ and $a_{v_2}=\frac{3}{2}$. However, if $v_1$ and $v_2$ swap their incoming contracts from $u_1$ and $u_2$, then the new payments result in assets of $a_{v_1}=1$ and $a_{v_2}=\frac{3}{2}$, which is a strict improvement for $v_1$.
\end{proof}

Note that since the assets of both $v_1$ and $v_2$ are non-decreasing in a semi-positive swap, one can use the monotonicity property to show that such a semi-positive swap is also acceptable to every other bank in the system, i.e. all banks $w \in B$ will have $a_w' \geq a_w$.

However, arguing about a positive swap is a more difficult question. This is especially true if we consider a general network topology with cycles; as the topic of portfolio compression indicates, directed cycles are indeed often present and play an important role in financial networks in practice \cite{cycles, cycles2}. Once we have a cyclic network, a debt swap can significantly reorganize the network, e.g. it can create new cycles, and reconnect or remove old ones. Such a change can lead to a significantly different network configuration, and hence a very different solution than in the initial state.

For an example, consider the system in Figure \ref{fig:invariants}, where the acting nodes have $a_{v_1}=\frac{1}{2}$, $a_{v_2}=1$ before, and $a'_{v_1}=\frac{3}{4}$, $a'_{v_2}=\frac{1}{2}$ after the swap. One can observe that there are no obvious invariants in the system: the sum of assets of acting nodes, the sum of assets of all banks and the total payment on all contracts are all changing due to the operation. As such, the creation and removal of cycles can result in a very different solution for the system, and hence analyzing if a positive swap exists is a more challenging task in general.

\begin{figure}
\centering
	\definecolor{dgray}{gray}{0.35}

\begin{tikzpicture}

	\draw[very thick, dgray, arrows=-latex] (0pt,60pt) -- (0pt,8pt);
	\draw[very thick, dgray, arrows=-latex] (30pt,60pt) -- (30pt,8pt);
	\draw[very thick, dgray, arrows=-latex] (30pt,0pt) -- (51pt,25pt);
	\draw[very thick, dgray, arrows=-latex] (55pt,30pt) -- (50pt,60pt) -- (38pt,60pt);

	\draw[black, fill=white] (0pt,0pt) circle (8pt);
	\draw[black, fill=white] (0pt,60pt) circle (8pt);
	\draw[black, fill=white] (30pt,0pt) circle (8pt);
	\draw[black, fill=white] (30pt,60pt) circle (8pt);
	\draw[black, fill=white] (55pt,30pt) circle (8pt);
	
	\node[anchor=center] at (0.5pt,-0.5pt) {\normalsize $v_1$};
	\node[anchor=center] at (0.5pt,59.5pt) {\normalsize $u_1$};
	\node[anchor=center] at (30.5pt,-0.5pt) {\normalsize $v_2$};
	\node[anchor=center] at (30.5pt,59.5pt) {\normalsize $u_2$};
	
	\node[anchor=center] at (-6pt,33pt) {\scriptsize $\left[^{\!}\frac{1}{2}^{\!}\right]$};
	\node[anchor=center] at (24pt,33pt) {\scriptsize $[^{^{\!}}1^{^{\!}}]$};
	\node[anchor=center] at (46pt,9.5pt) {\scriptsize $[^{^{\!}}1^{^{\!}}]$};
	\node[anchor=center] at (57pt,52pt) {\scriptsize $[^{^{\!}}1^{^{\!}}]$};
	
	
	\draw [fill=white] (3.5pt,56pt) rectangle (9.5pt,44pt);
	\node[anchor=center] at (6.5pt,50pt) {\footnotesize $\frac{1}{2}$};
	\draw [fill=white] (33.5pt,56pt) rectangle (39.5pt,44pt);
	\node[anchor=center] at (36.5pt,50pt) {\footnotesize $\frac{1}{4}$};
	
	
	\draw[line width=2.5pt, arrows=-stealth] (75pt,30pt) -- (95pt,30pt);
	
	
	\draw[very thick, dgray, arrows=-latex] (115pt,60pt) -- (143pt,7pt);
	\draw[very thick, dgray, arrows=-latex] (145pt,60pt) -- (117pt,7pt);
	\draw[very thick, dgray, arrows=-latex] (145pt,0pt) -- (166pt,25pt);
	\draw[very thick, dgray, arrows=-latex] (170pt,30pt) -- (165pt,60pt) -- (153pt,60pt);

	\draw[black, fill=white] (115pt,0pt) circle (8pt);
	\draw[black, fill=white] (115pt,60pt) circle (8pt);
	\draw[black, fill=white] (145pt,0pt) circle (8pt);
	\draw[black, fill=white] (145pt,60pt) circle (8pt);
	\draw[black, fill=white] (170pt,30pt) circle (8pt);
	
	\node[anchor=center] at (115.5pt,-0.5pt) {\normalsize $v_1$};
	\node[anchor=center] at (115.5pt,59.5pt) {\normalsize $u_1$};
	\node[anchor=center] at (145.5pt,-0.5pt) {\normalsize $v_2$};
	\node[anchor=center] at (145.5pt,59.5pt) {\normalsize $u_2$};
	
	\node[anchor=center] at (141.5pt,25.5pt) {\scriptsize $\left[^{\!}\frac{1}{2}^{\!}\right]$};
	\node[anchor=center] at (118.5pt,25.5pt) {\scriptsize $\left[^{\!}\frac{3}{4}^{\!}\right]$};
	\node[anchor=center] at (162.5pt,9pt) {\scriptsize $\left[^{\!}\frac{1}{2}^{\!}\right]$};
	\node[anchor=center] at (173pt,52.5pt) {\scriptsize $\left[^{\!}\frac{1}{2}^{\!}\right]$};
	
	
	\draw [fill=white] (118.5pt,56pt) rectangle (124.5pt,44pt);
	\node[anchor=center] at (121.5pt,50pt) {\footnotesize $\frac{1}{2}$};
	\draw [fill=white] (148.5pt,56pt) rectangle (154.5pt,44pt);
	\node[anchor=center] at (151.5pt,50pt) {\footnotesize $\frac{1}{4}$};
	
\end{tikzpicture}
	\caption{Example for debt swapping in a network topology with cycles.}
	\label{fig:invariants}
\end{figure}
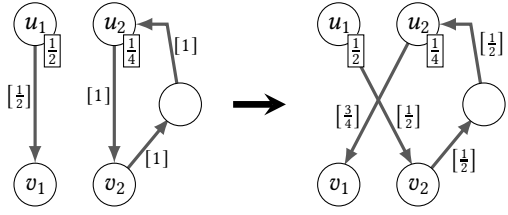

\subsection{No positive swap in the base model} \label{sec:nopos}

However, even with the reorganization of cycles, it turns out that we cannot have a positive swap in the base model. Proving this statement is more challenging than it seems at first glance; we outline the base idea of the proof below, and we discuss the technical details in Appendix \ref{App:B}.

\begin{theorem} \label{th:nopos_base}
There is no positive swap in the base model.
\end{theorem}

For the proof, let us use $p_1:=p_{u_1, v_1}$ and $p_2:=p_{u_2, v_2}$ to denote the payments on the swapped contracts in the initial state, and $p'_1:=p'_{u_2, v_1}$ and $p'_2:=p'_{u_1, v_2}$ to denote them after the swap.

Instead of analyzing a financial network directly, we will look at a slightly modified version by breaking up the cycles that contain the swapped contracts. That is, we create a so-called \textit{open variant} of the financial network: we remove the contracts $(u_1, v_1)$ and $(u_2, v_2)$, and instead we add a source $s_1$ that has a single liability $l_{s_1, v_1}=\infty$, and a source $s_2$ with a single liability $l_{s_2, v_2}=\infty$. Similarly, we redirect the outgoing side of the swapped contracts to two new target nodes $t_1$ and $t_2$ that only have a single incoming debt, i.e. $l_{u_1, t_1}=l_{u_1, v_1}$ and $l_{u_2, t_2}=l_{u_2, v_2}$.

We know that in the original (\textit{closed}) system, the payments $p_1, p_2$ and the payments $p_1', p_2'$ provide the maximal solution of the system before and after the swap, respectively. This implies that if we set $e_{s_1}=p_1$ and $e_{s_2}=p_2$, then the resulting solution in the open system must have $a_{t_1}=p_1$ and $a_{t_1}=p_2$. Similarly, if we set $e_{s_1}=p_1'$ and $e_{s_2}=p_2'$, then the resulting solution must have $a_{t_1}=p_2'$ and $a_{t_2}=p_1'$.  We illustrate this connection between the closed and open system in Figure \ref{fig:openclosed}.

Through these relations, the open variant allows us to study the closed network with the properties we have established for source and sink nodes. In the open system, the network topology does not have to be modified at all; the state after swapping simply corresponds to a different choice of $e_{s_1}$ and $e_{s_2}$.

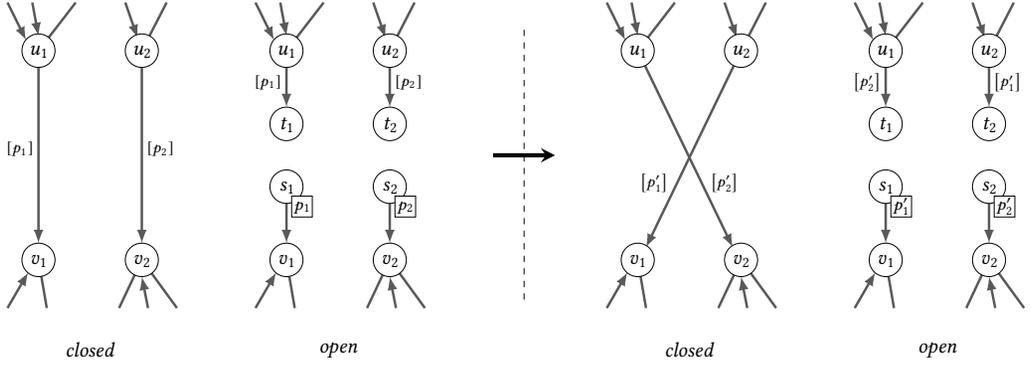
\begin{figure}
\centering
	\resizebox{1.0\textwidth}{!}{\definecolor{dgray}{gray}{0.35}
\definecolor{lgray}{gray}{0.8}

\begin{tikzpicture}

	\draw[very thick, dgray] (0pt,100pt) -- (18pt,123pt);
	\draw[very thick, dgray, arrows=-latex] (-15pt,123pt) -- (-6pt,105pt);
	\draw[very thick, dgray, arrows=-latex] (-3pt,123pt) -- (0pt,107pt);
	
	\draw[very thick, dgray] (50pt,100pt) -- (62pt,123pt);
	\draw[very thick, dgray, arrows=-latex] (42pt,123pt) -- (48pt,106pt);
	
	\draw[very thick, dgray, arrows=-latex] (-15pt,-23pt) -- (-5pt,-6pt);
	\draw[very thick, dgray] (0pt,0pt) -- (4pt,-23pt);
	
	\draw[very thick, dgray] (50pt,0pt) -- (39pt,-23pt);
	\draw[very thick, dgray, arrows=-latex] (53pt,-23pt) -- (50pt,-8pt);
	\draw[very thick, dgray] (50pt,0pt) -- (67pt,-23pt);

	\draw[very thick, dgray, arrows=-latex] (0pt,100pt) -- (0pt,8pt);
	\draw[very thick, dgray, arrows=-latex] (50pt,100pt) -- (50pt,8pt);
	
	\draw[black, fill=white] (0pt,100pt) circle (8pt);
	\draw[black, fill=white] (50pt,100pt) circle (8pt);
	
	\node[anchor=center] at (0.5pt,99.5pt) {\normalsize $u_1$};
	\node[anchor=center] at (50.5pt,99.5pt) {\normalsize $u_2$};
	
	\draw[black, fill=white] (0pt,0pt) circle (8pt);
	\draw[black, fill=white] (50pt,0pt) circle (8pt);
	
	\node[anchor=center] at (0.5pt,-0.5pt) {\normalsize $v_1$};
	\node[anchor=center] at (50.5pt,-0.5pt) {\normalsize $v_2$};
	
	\node[anchor=center] at (25pt,-43pt) {\normalsize \textit{closed}};
	
	\node[anchor=center] at (-9pt,53pt) {\scriptsize $\left[ p_1 \right]$};
	\node[anchor=center] at (59pt,53pt) {\scriptsize $\left[ p_2 \right]$};
	
	
	\draw[very thick, dgray] (120pt,100pt) -- (138pt,123pt);
	\draw[very thick, dgray, arrows=-latex] (105pt,123pt) -- (114pt,105pt);
	\draw[very thick, dgray, arrows=-latex] (117pt,123pt) -- (120pt,107pt);
	
	\draw[very thick, dgray] (170pt,100pt) -- (182pt,123pt);
	\draw[very thick, dgray, arrows=-latex] (162pt,123pt) -- (168pt,106pt);
	
	\draw[very thick, dgray, arrows=-latex] (105pt,-23pt) -- (115pt,-6pt);
	\draw[very thick, dgray] (120pt,0pt) -- (124pt,-23pt);
	
	\draw[very thick, dgray] (170pt,0pt) -- (159pt,-23pt);
	\draw[very thick, dgray, arrows=-latex] (173pt,-23pt) -- (170pt,-8pt);
	\draw[very thick, dgray] (170pt,0pt) -- (187pt,-23pt);
	
	\draw[very thick, dgray, arrows=-latex] (120pt,100pt) -- (120pt,73pt);
	\draw[very thick, dgray, arrows=-latex] (170pt,100pt) -- (170pt,73pt);
	\draw[very thick, dgray, arrows=-latex] (120pt,30pt) -- (120pt,8pt);
	\draw[very thick, dgray, arrows=-latex] (170pt,30pt) -- (170pt,8pt);
	
	\draw[black, fill=white] (120pt,100pt) circle (8pt);
	\draw[black, fill=white] (170pt,100pt) circle (8pt);
	
	\node[anchor=center] at (120.5pt,99.5pt) {\normalsize $u_1$};
	\node[anchor=center] at (170.5pt,99.5pt) {\normalsize $u_2$};
	
	\draw[black, fill=white] (120pt,65pt) circle (8pt);
	\draw[black, fill=white] (170pt,65pt) circle (8pt);
	
	\node[anchor=center] at (120.5pt,64.5pt) {\normalsize $t_1$};
	\node[anchor=center] at (170.5pt,64.5pt) {\normalsize $t_2$};
	
	\draw[black, fill=white] (120pt,0pt) circle (8pt);
	\draw[black, fill=white] (170pt,0pt) circle (8pt);
	
	\node[anchor=center] at (120.5pt,-0.5pt) {\normalsize $v_1$};
	\node[anchor=center] at (170.5pt,-0.5pt) {\normalsize $v_2$};
	
	\draw[black, fill=white] (120pt,35pt) circle (8pt);
	\draw[black, fill=white] (170pt,35pt) circle (8pt);
	
	\node[anchor=center] at (120.5pt,34.5pt) {\normalsize $s_1$};
	\node[anchor=center] at (170.5pt,34.5pt) {\normalsize $s_2$};
	
	\node[anchor=center] at (145pt,-43pt) {\normalsize \textit{open}};
	
	\draw [fill=white] (122.5pt,30.5pt) rectangle (132.5pt,20.5pt);
	\node[anchor=center] at (127.75pt,25pt) {\footnotesize $p_1$};
	\draw [fill=white] (172.5pt,30.5pt) rectangle (182.5pt,20.5pt);
	\node[anchor=center] at (177.75pt,25pt) {\footnotesize $p_2$};
	
	\node[anchor=center] at (111pt,85pt) {\scriptsize $\left[ p_1 \right]$};
	\node[anchor=center] at (179pt,85pt) {\scriptsize $\left[ p_2 \right]$};


	\draw[dashed] (235pt,110pt) -- (235pt,-20pt);
	\draw[line width=2pt, arrows=-stealth] (220pt,50pt) -- (250pt,50pt);
	
	
	\draw[very thick, dgray] (290pt,100pt) -- (308pt,123pt);
	\draw[very thick, dgray, arrows=-latex] (275pt,123pt) -- (284pt,105pt);
	\draw[very thick, dgray, arrows=-latex] (287pt,123pt) -- (290pt,107pt);
	
	\draw[very thick, dgray] (340pt,100pt) -- (352pt,123pt);
	\draw[very thick, dgray, arrows=-latex] (332pt,123pt) -- (338pt,106pt);
	
	\draw[very thick, dgray, arrows=-latex] (275pt,-23pt) -- (285pt,-6pt);
	\draw[very thick, dgray] (290pt,0pt) -- (294pt,-23pt);
	
	\draw[very thick, dgray] (340pt,0pt) -- (329pt,-23pt);
	\draw[very thick, dgray, arrows=-latex] (343pt,-23pt) -- (340pt,-8pt);
	\draw[very thick, dgray] (340pt,0pt) -- (357pt,-23pt);
	
	\draw[very thick, dgray, arrows=-latex] (290pt,100pt) -- (336pt,6pt);
	\draw[very thick, dgray, arrows=-latex] (340pt,100pt) -- (294pt,6pt);
	
	\draw[black, fill=white] (290pt,100pt) circle (8pt);
	\draw[black, fill=white] (340pt,100pt) circle (8pt);
	
	\node[anchor=center] at (290.5pt,99.5pt) {\normalsize $u_1$};
	\node[anchor=center] at (340.5pt,99.5pt) {\normalsize $u_2$};
	
	\draw[black, fill=white] (290pt,0pt) circle (8pt);
	\draw[black, fill=white] (340pt,0pt) circle (8pt);
	
	\node[anchor=center] at (290.5pt,-0.5pt) {\normalsize $v_1$};
	\node[anchor=center] at (340.5pt,-0.5pt) {\normalsize $v_2$};
	
	\node[anchor=center] at (315pt,-43pt) {\normalsize \textit{closed}};
	
	\node[anchor=center] at (298pt,37pt) {\scriptsize $\left[ p_1' \right]$};
	\node[anchor=center] at (332pt,37pt) {\scriptsize $\left[ p_2' \right]$};
	
	
	\draw[very thick, dgray] (410pt,100pt) -- (428pt,123pt);
	\draw[very thick, dgray, arrows=-latex] (395pt,123pt) -- (404pt,105pt);
	\draw[very thick, dgray, arrows=-latex] (407pt,123pt) -- (410pt,107pt);
	
	\draw[very thick, dgray] (460pt,100pt) -- (472pt,123pt);
	\draw[very thick, dgray, arrows=-latex] (452pt,123pt) -- (458pt,106pt);
	
	\draw[very thick, dgray, arrows=-latex] (395pt,-23pt) -- (405pt,-6pt);
	\draw[very thick, dgray] (410pt,0pt) -- (414pt,-23pt);
	
	\draw[very thick, dgray] (460pt,0pt) -- (449pt,-23pt);
	\draw[very thick, dgray, arrows=-latex] (463pt,-23pt) -- (460pt,-8pt);
	\draw[very thick, dgray] (460pt,0pt) -- (477pt,-23pt);
	
	\draw[very thick, dgray, arrows=-latex] (410pt,100pt) -- (410pt,73pt);
	\draw[very thick, dgray, arrows=-latex] (460pt,100pt) -- (460pt,73pt);
	\draw[very thick, dgray, arrows=-latex] (410pt,30pt) -- (410pt,8pt);
	\draw[very thick, dgray, arrows=-latex] (460pt,30pt) -- (460pt,8pt);
	
	\draw[black, fill=white] (410pt,100pt) circle (8pt);
	\draw[black, fill=white] (460pt,100pt) circle (8pt);
	
	\node[anchor=center] at (410.5pt,99.5pt) {\normalsize $u_1$};
	\node[anchor=center] at (460.5pt,99.5pt) {\normalsize $u_2$};
	
	\draw[black, fill=white] (410pt,65pt) circle (8pt);
	\draw[black, fill=white] (460pt,65pt) circle (8pt);
	
	\node[anchor=center] at (410.5pt,64.5pt) {\normalsize $t_1$};
	\node[anchor=center] at (460.5pt,64.5pt) {\normalsize $t_2$};
	
	\draw[black, fill=white] (410pt,0pt) circle (8pt);
	\draw[black, fill=white] (460pt,0pt) circle (8pt);
	
	\node[anchor=center] at (410.5pt,-0.5pt) {\normalsize $v_1$};
	\node[anchor=center] at (460.5pt,-0.5pt) {\normalsize $v_2$};
	
	\draw[black, fill=white] (410pt,35pt) circle (8pt);
	\draw[black, fill=white] (460pt,35pt) circle (8pt);
	
	\node[anchor=center] at (410.5pt,34.5pt) {\normalsize $s_1$};
	\node[anchor=center] at (460.5pt,34.5pt) {\normalsize $s_2$};
	
	\node[anchor=center] at (435pt,-43pt) {\normalsize \textit{open}};
	
	\draw [fill=white] (412.5pt,30.5pt) rectangle (422.5pt,20.5pt);
	\node[anchor=center] at (417.75pt,25.5pt) {\footnotesize $p_1'$};
	\draw [fill=white] (462.5pt,30.5pt) rectangle (472.5pt,20.5pt);
	\node[anchor=center] at (467.75pt,25.5pt) {\footnotesize $p_2'$};
	
	\node[anchor=center] at (401pt,85pt) {\scriptsize $\left[ p_2' \right]$};
	\node[anchor=center] at (469pt,85pt) {\scriptsize $\left[ p_1' \right]$};
	
\end{tikzpicture}}
	\caption{Relationships between the swapping banks in the closed and open version of the system, before swapping (left) and after swapping (right). Liabilities are not shown for simplicity.}
	\label{fig:openclosed}
\end{figure}

For the proof of Theorem \ref{th:nopos_base}, let us assume that there is a positive swap in the network, and let us consider separate cases according to the payments on the swapped contracts, i.e. the relations between $p_1$, $p_2$, $p_1'$ and $p_2'$. Note that a positive swap does not necessarily mean that $p_1' > p_1$ and $p_2' > p_2$, since the banks $v_1$ and $v_2$ also receive payments from other parts of the network. As such, it could also happen in a positive swap that $p_1' < p_1$, but we still have $a_{v_1}' > a_{v_1}$ because $v_1$ receives more payment on some other debt due to increased assets of $v_2$. 

\begin{lemma} \label{lem:monoton}
We cannot have a positive swap where $p_1' \leq p_1$ and $p_2' \leq p_2$.
\end{lemma}

\begin{proof}
If we have $p_1' \leq p_1$ and $p_2' \leq p_2$, then both $s_1$ and $s_2$ have less funds in the second state of the open system (i.e. after swapping). Then monotonicity implies that we cannot have $a_{v_1}' > a_{v_1}$ or $a_{v_2}' > a_{v_2}$.
\end{proof}

Hence for a positive swap, at least one of the acting nodes must receive more payment on the swapped contract after swapping. Assume w.l.o.g. that this is $v_2$, i.e. that $p_2'>p_2$. Based on the payment on the other contract, we consider the cases $p_1' > p_1$ and $p_1' \leq p_1$ separately.

\begin{lemma} \label{lem:bothinc}
We cannot have a positive swap where $p_1' > p_1$ and $p'_2 > p_2$.
\end{lemma}

\renewcommand{\proofname}{Proof sketch.}

\begin{proof}
Consider the two states of the open system: for $e_{s_1}=p_1$, $e_{s_2}=p_2$ we have $a_{t_1}=p_1$ and $a_{t_2}=p_2$ (original state), and for $e_{s_1}=p_1'$, $e_{s_2}=p_2'$, we have $a_{t_1}=p_2'$ and $a_{t_2}=p_1'$ (after swapping). Between these two states, the funds at the sources and the assets at the sinks both increase by $(p_1'+p_2')-(p_1+p_2)$, so both $s_1$ and $s_2$ in the open system are $(t_1,t_2)$-linear (on $[p_1, p_1']$ and $[p_2, p_2']$, respectively).

Hence there exists $\alpha_1, \beta_1 \in [0,1]$ with $\alpha_1 + \beta_1 = 1$, such that for sufficiently small $\delta>0$, setting $e_{s_1}=a+\delta$ results in $a_{t_1}=p_1+\alpha_1 \cdot \delta$ and $a_{t_2}=p_2 + \beta_1 \cdot \delta$. Also, we cannot have $\alpha_1=1, \beta_1=0$; in this case, setting $e_{s_1}=p_1+\delta$ would result in $a_{t_1}=p_1 + \delta$, which gives a larger solution in the original system with $p_{u_1, v_1}=p_1+ \delta$, contradicting the fact that our initial clearing vector was maximal. Note that this argument also uses the fact that $l_{u_1, v_1} \geq p_1+ \delta$, which indeed holds for a small enough $\delta$; we discuss this in Appendix \ref{App:B}.

Similarly, there exists $\alpha_2, \beta_2 \in [0,1]$ with $\alpha_2 + \beta_2 = 1$, such that setting $e_{s_2}=p_2+\delta$ results in $a_{t_1}=p_1+\alpha_2 \cdot \delta$ and $a_{t_2}=p_2 + \beta_2 \cdot \delta$, and we must have $\alpha_2>0$.

Now let us select a small enough $\delta$, define $\Delta_1 = \delta / \beta_1$ and $\Delta_2 = \delta / \alpha_2$, and consider the increase $e_{s_1}=p_1 + \Delta_1$, $e_{s_2}=p_2 + \Delta_2$. This provides a higher solution in the original closed system, since
\begin{displaymath}
a_{t_1}=p_1 + \alpha_1 \cdot \Delta_1 + \alpha_2 \cdot \Delta_2 = p_1 + \delta \cdot \left( \frac{\alpha_1}{\beta_1} + 1 \right) = p_1 + \frac{\delta}{\beta_1} =p_1 + \Delta_1 \, ,
\end{displaymath}
\begin{displaymath}
 a_{t_2}=p_2 + \beta_1 \cdot \Delta_1 + \beta_2 \cdot \Delta_2 = p_2 + \delta \cdot \left( 1 + \frac{\beta_2}{\alpha_2}\right) = p_2 + \frac{\delta}{\alpha_2} = p_2 + \Delta_2 \, .
\end{displaymath}
This contradicts the fact that our original clearing was maximal.
\end{proof}

The only case that remains is when $p_1' \leq p_1$, i.e. when $v_1$ receives less money in the network through the swapped contract, but it is still compensated for this loss by a larger payment on some other incoming debts. The formal proof of this claim is significantly more technical; we outline the main idea here, and discuss the details in Appendix \ref{App:B}.

\begin{lemma} \label{lem:hard}
We cannot have a positive swap where $p_1' \leq p_1$ and $p'_2 > p_2$.
\end{lemma}

\begin{proof}
The proof idea is similar to above: we prove that even in this case, it still holds that both $s_1$ and $s_2$ are $(t_1,t_2)$-linear on a small interval above $p_1$ and $p_2$. This allows us to use the same argument as before; the challenging part is to prove $(t_1,t_2)$-linearity in this slightly different setting.

Intuitively, the proof works as follows: one can show that if $a_{v_1}' > a_{v_1}$, then when we increase $e_{s_2}$ from $p_2$ to $p_2'$, there must exist a ``flow'' of strictly more than $p_1-p_1'$ new assets from $s_2$ to $v_1$. Note that this is a non-trivial claim: an increase of $p_1-p_1'$ in $a_{v_1}$ could also be caused by less than $p_1-p_1'$ new funds, as in Figure \ref{fig:expansive}, since $v_1$ is not a sink node.
 
Assume there is a such a flow of $p_1-p_1'+\epsilon$ for some $\epsilon>0$. If we subtract the path of these $p_1-p_1'+\epsilon$ assets, then the remaining $(p_2'-p_2)-(p_1-p_1'+\epsilon) = (p_1'+p_2')-(p_1+p_2)-\epsilon$ new funds at $s_2$ can only contribute an increase of $(p_1'+p_2')-(p_1+p_2)-\epsilon$ to the sinks $t_1$ and $t_2$ altogether, due to non-expansivity.

However, the total increase of assets at the sink nodes is $(p_1'+p_2')-(p_1+p_2)$ between the two states of the open system, so the new assets going from $s_2$ to $v_1$ must also contribute at least $\epsilon$ to the sink nodes. The ``first'' $p_1-p_1'$ assets from $s_2$ to $v_1$ cannot have any such contribution, since they have (at most) the same effect as increasing $e_{s_1}$ from $p_1'$ to $p_1$; however, $s_1$ already had these extra funds in the original state, and thus they are already included in our initial state of $a_{t_1}=p_1$, $a_{t_2}=p_2$. Hence the remaining flow of $\epsilon$ assets must contribute $\epsilon$ to the sinks; however, these assets have the same effect as increasing $e_{s_1}$ to $p_1+\epsilon$ in the original state. This means that $s_1$ is $(t_1,t_2)$-linear on $[p_1, p_1+\epsilon]$; due to concavity, it is also $(t_1,t_2)$-linear on $[0, p_1+\epsilon]$.

Since the increase of $p_2'-p_2$ at $e_{s_2}$ has provided $(p_1'+p_2')-(p_1+p_2)-\epsilon$ new assets to the sinks directly, and $p_1-p_1'+\epsilon$ new assets to a $(t_1,t_2)$-linear node $v_1$, it provides altogether $p_2'-p_2$ new assets to the sinks; hence $s_{v_2}$ is also $(t_1,t_2)$-linear on $[p_2, p_2']$.
\end{proof}

\renewcommand{\proofname}{Proof.}

This concludes the proof of Theorem \ref{th:nopos_base}. Note that the theorem also has immediate implications on the more sophisticated operation of portfolio swapping. 

\begin{theorem} \label{th:noport_base}
There is no positive portfolio swap in the base model.
\end{theorem}

\begin{proof}
Given a positive portfolio swap from debtor sets $U_1$ and $U_2$, we could introduce two auxiliary nodes $\hat{v_1}$ and $\hat{v_2}$ that simply collect and relay these contracts: for all $u \in U_i$ (with $i \in \{1,2\}$), we replace the debt $(u, v_i)$ by a debt $(u, \hat{v_i})$, and we assign $e_{\hat{v_i}}=0$, $l_{\hat{v_i}, v_i}=\infty$. The portfolio swap in the original network now provides a positive swap in this system, which contradicts Theorem \ref{th:nopos_base}. 
\end{proof}

Finally, we point out that positive swaps can easily become possible in more complex models of financial networks. For example, one popular extension of the model is to consider a default cost parameter $\beta \in (0,1)$, and assume that whenever a bank $v$ is in default, then its assets are multiplied by this reduction factor $\beta$ to account for the administrative costs of a default in practice. If we consider this extended model with a choice of $\beta=\frac{1}{2}$, then Figure \ref{fig:semipos} already turns into an example of a positive swap: one can compute that it provides $a_{v_1}=\frac{1}{4}$ and $a_{v_2}=\frac{9}{8}$, but $a_{v_1}'=1$, $a_{v_2}'=\frac{5}{4}$.

\subsection{Implications for the proportional shock model} \label{sec:prop}

These results in the base model have immediate implications for the case of the proportional shock model.

\begin{lemma}
There is no swap that is positive for a specific $\lambda$ in the proportional shock model, i.e. a swap that provides $f'_{v_1}(\lambda)>f_{v_1}(\lambda)$ and $f'_{v_2}(\lambda)>f_{v_2}(\lambda)$.
\end{lemma}

\renewcommand{\proofname}{Proof.}

\begin{proof}
Given an example for such a swap, multiplying the funds of all nodes by $(1-\lambda)$ would give us a network that contradicts Theorem \ref{th:nopos_base} in the base model. 
\end{proof}

However, this does not immediately imply that there is no positive swap in the proportional model at all. Since semi-positive swaps are possible in the base model, it could still happen that a swap is positive because it provides $f'_{v_1}(\lambda_1)>f_{v_1}(\lambda_1)$ and $f'_{v_2}(\lambda_1)=f_{v_2}(\lambda_1)$ for a specific $\lambda_1$, and $f'_{v_1}(\lambda_2)=f_{v_1}(\lambda_2)$ and $f'_{v_2}(\lambda_2)>f_{v_2}(\lambda_2)$ for some other $\lambda_2$. However, one can show that this is also not possible, and hence there is no positive swap in the proportional model at all.

\renewcommand{\proofname}{Proof sketch.}

\begin{theorem} \label{th:nopos_prop}
There is no positive swap in the proportional shock model.
\end{theorem}

\begin{proof}
Similarly to the proof of Theorem \ref{th:nopos_base}, one can show that a semi-positive swap also implies that specific linearities must hold. In particular, if we have $p_1' = p_1$ besides $p_2' > p_2$, then $s_2$ must be $(t_1,t_2)$-linear on $[p_2, p_2']$; and if we have $p_1' < p_1$ besides $p_2' > p_2$, then a ``flow'' of $p_1-p_1'$ funds must go directly to $v_1$, and the remaining part of the new funds is $(t_1,t_2)$-linear.

Due to our semi-positivity assumptions, these properties must hold for bank $s_1$ when we multiply the funds of each bank by $(1-\lambda_1)$, and must hold for $s_2$ when we multiply the funds by $(1-\lambda_2)$. Furthermore, one can observe that these linearities also remain true if we scale down the funds of all banks in a system; as such, if we have e.g. $\lambda_1<\lambda_2$, then the property established for $\lambda_1$ must also hold for $\lambda_2$, since the shock of size $\lambda_2$ can be obtained from the shock of size $\lambda_1$ with a further scaling of $\frac{1-\lambda_2}{1-\lambda_1}$.

This means that in the financial network with a shock of size $\lambda_2$, both of these properties hold: $s_1$ is $(t_1,t_2)$-linear except for (possibly) a flow towards $v_2$, and $s_2$ is $(t_1,t_2)$-linear except for (possibly) a flow towards $v_1$. Hence intuitively, a small $\delta$ amount of extra funds injected at $e_{s_1}$ (or $e_{s_2}$) will all arrive at $t_1$, $t_2$ and $v_2$ (or $t_1$, $t_2$ and $v_1$, respectively). Similarly to the proof of Lemma \ref{lem:bothinc}, this allows us to present a slightly larger clearing vector in the system, which is a contradiction.
\end{proof}

\renewcommand{\proofname}{Proof.}

\noindent Similarly to Theorem \ref{th:noport_base}, this negative result can be carried over to the case of portfolio swapping.

\section{Worst-case shock models} \label{sec:worstcase}

We now analyze swapping operations in the worst-case shock models. Our example in Section \ref{sec:motive} has already shown that we can indeed have swaps in these models that are beneficial for both acting nodes. However, this raises a range of interesting questions to study, e.g. about the effects of the swaps on the rest of the system, or the algorithmic aspects of finding such swaps.

In this section, we focus on the properties of the worst-set shock model. The worst-sum model is similar to this in many aspects, so most of our results also carry over to that model; however, the proofs in this case become somewhat more technical. We discuss the details of adapting our proofs to the worst-sum model in Appendix \ref{App:D}.

\subsection{Effect on other banks}

We first explore the most fundamental properties of swapping in this worst-case setting. In particular, we show that even though a positive swap is beneficial for both of the swapping parties, it can result in a strictly worse situation for some other banks in the network.

\begin{lemma} \label{lem:badforw}
If two acting nodes $v_1$, $v_2$ execute a positive swap, then this can result in a new shock function $f_w' < f_w$ for some other bank $w$.
\end{lemma}

\begin{proof}
Consider Figure \ref{fig:badforw}, which is an extension of our motivational example by two new banks $u_4$ and $w$. The modifications do not affect the worst-set shock function in the original part of the system, so the previously studied swap is still positive for $v_1$, $v_2$. The new contracts after swapping are now only indicated by dashed arrows in the figure.

Regardless of the swap, the new bank $w$ has $f_w(0)=f'_w(0)=4$ and $f_w(2)=f'_w(2)=0$, so the only interesting case is when $k=1$. Before the swap, a shock at either of the source nodes results in a loss of $2$, so $f_w(1)=2$. However, after $v_1$ and $v_2$ swaps, a shock at $s_2$ (setting $e_{s_2}'=0$) means that $w$ will only receive $1$ unit of money through $v_1$, so $f'_w(1)=1$.
\end{proof}

\begin{figure}
\hspace{0.01\textwidth}
\minipage{0.42\textwidth}
	\centering
	\definecolor{dgray}{gray}{0.35}

\begin{tikzpicture}

	\draw[very thick, dgray, arrows=-latex] (0pt,15pt) -- (28pt,2pt);
	\draw[very thick, dgray, arrows=-latex] (0pt,15pt) -- (28pt,28pt);
	\draw[very thick, dgray, arrows=-latex] (0pt,70pt) -- (28pt,57pt);
	\draw[very thick, dgray, arrows=-latex] (0pt,70pt) -- (28pt,83pt);
	\draw[very thick, dgray, arrows=-latex] (35pt,0pt) -- (73pt,18pt);
	\draw[very thick, dgray, arrows=-latex] (35pt,30pt) -- (73pt,26pt);
	\draw[very thick, dgray, arrows=-latex] (35pt,55pt) -- (73pt,59pt);
	\draw[very thick, dgray, arrows=-latex] (35pt,85pt) -- (73pt,67pt);
	
	\draw[dashed, arrows=-latex] (35pt,30pt) -- (75pt,57pt);
	\draw[dashed, arrows=-latex] (35pt,55pt) -- (75pt,28pt);
	
	\draw[very thick, dgray, arrows=-latex] (0pt,15pt) -- (28pt,-26pt);
	\draw[very thick, dgray, arrows=-latex] (80pt,63pt) -- (125pt,26pt);
	\draw[very thick, dgray, arrows=-latex] (35pt,-30pt) -- (75pt,-30pt) -- (125pt,18pt);

	\draw[black, fill=white] (0pt,15pt) circle (8pt);
	\draw[black, fill=white] (0pt,70pt) circle (8pt);
	\draw[black, fill=white] (35pt,0pt) circle (8pt);
	\draw[black, fill=white] (35pt,30pt) circle (8pt);
	\draw[black, fill=white] (35pt,55pt) circle (8pt);
	\draw[black, fill=white] (35pt,85pt) circle (8pt);
	\draw[black, fill=white] (80pt,22pt) circle (8pt);
	\draw[black, fill=white] (80pt,63pt) circle (8pt);
	
	\draw[black, fill=white] (35pt,-30pt) circle (8pt);
	\draw[black, fill=white] (130pt,22pt) circle (8pt);
	
	\node[anchor=center] at (0.5pt,14.5pt) {\normalsize $s_2$};
	\node[anchor=center] at (0.5pt,69.5pt) {\normalsize $s_1$};
	\node[anchor=center] at (35.5pt,-0.5pt) {\normalsize $u_3$};
	\node[anchor=center] at (35.5pt,29.5pt) {\normalsize $u_2$};
	\node[anchor=center] at (35.5pt,54.5pt) {\normalsize $u_1$};
	\node[anchor=center] at (35.5pt,84.5pt) {\normalsize $u_0$};
	\node[anchor=center] at (80.5pt,21.5pt) {\normalsize $v_2$};
	\node[anchor=center] at (80.5pt,62.5pt) {\normalsize $v_1$};
	
	\node[anchor=center] at (35.5pt,-30.5pt) {\normalsize $u_4$};
	\node[anchor=center] at (130.5pt,21.5pt) {\normalsize $w$};
	
	\node[anchor=center] at (104pt,51pt) {\footnotesize $2$};
	\node[anchor=center] at (104pt,-10pt) {\footnotesize $2$};
	\node[anchor=center] at (12pt,-12pt) {\footnotesize $2$};
	
	\draw [fill=white] (3.5pt,11.5pt) rectangle (9.5pt,1.5pt);
	\node[anchor=center] at (6.5pt,7pt) {\small $4$};
	\draw [fill=white] (3.5pt,66.5pt) rectangle (9.5pt,56.5pt);
	\node[anchor=center] at (6.5pt,61.5pt) {\small $4$};

\end{tikzpicture}
	\caption{Example when a positive swap creates a worse situation for a bank $w$.}
	\label{fig:badforw}
\endminipage\hfill
\hspace{0.09\textwidth}
\minipage{0.45\textwidth}
	\centering
	\vspace{15pt}
	\definecolor{dgray}{gray}{0.35}

\begin{tikzpicture}

	\draw[very thick, dgray, arrows=-latex] (0pt,15pt) -- (28pt,2pt);
	\draw[very thick, dgray, arrows=-latex] (0pt,15pt) -- (28pt,28pt);
	\draw[very thick, dgray, arrows=-latex] (0pt,70pt) -- (28pt,57pt);
	\draw[very thick, dgray, arrows=-latex] (0pt,70pt) -- (28pt,83pt);
	\draw[very thick, dgray, arrows=-latex] (35pt,0pt) -- (73pt,18pt);
	\draw[very thick, dgray, arrows=-latex] (35pt,30pt) -- (73pt,26pt);
	\draw[very thick, dgray, arrows=-latex] (35pt,55pt) -- (73pt,59pt);
	\draw[very thick, dgray, arrows=-latex] (35pt,85pt) -- (73pt,67pt);
	
	\draw[dashed, arrows=-latex] (35pt,30pt) -- (75pt,57pt);
	\draw[dashed, arrows=-latex] (35pt,55pt) -- (75pt,28pt);
	
	\draw[very thick, dgray, arrows=-latex] (80pt,63pt) -- (117pt,63pt);
	\draw[very thick, dgray, arrows=-latex] (80pt,63pt) -- (80pt,98pt) -- (75pt,103pt) -- (-25pt,103pt) -- (-30pt,98pt) -- (-30pt,20pt) -- (-25pt,15pt) -- (-8pt,15pt);

	\draw[black, fill=white] (0pt,15pt) circle (8pt);
	\draw[black, fill=white] (0pt,70pt) circle (8pt);
	\draw[black, fill=white] (35pt,0pt) circle (8pt);
	\draw[black, fill=white] (35pt,30pt) circle (8pt);
	\draw[black, fill=white] (35pt,55pt) circle (8pt);
	\draw[black, fill=white] (35pt,85pt) circle (8pt);
	\draw[black, fill=white] (80pt,22pt) circle (8pt);
	\draw[black, fill=white] (80pt,63pt) circle (8pt);
	
	\draw[black, fill=white] (125pt,63pt) circle (8pt);
	
	\node[anchor=center] at (0.5pt,14.5pt) {\normalsize $s_2$};
	\node[anchor=center] at (0.5pt,69.5pt) {\normalsize $s_1$};
	\node[anchor=center] at (35.5pt,-0.5pt) {\normalsize $u_3$};
	\node[anchor=center] at (35.5pt,29.5pt) {\normalsize $u_2$};
	\node[anchor=center] at (35.5pt,54.5pt) {\normalsize $u_1$};
	\node[anchor=center] at (35.5pt,84.5pt) {\normalsize $u_0$};
	\node[anchor=center] at (80.5pt,21.5pt) {\normalsize $v_2$};
	\node[anchor=center] at (80.5pt,62.5pt) {\normalsize $v_1$};
	
	\node[anchor=center] at (125.5pt,62.5pt) {\normalsize $t_1$};
	
	\node[anchor=center] at (100pt,69pt) {\footnotesize $2$};
	
	\draw [fill=white] (3.5pt,11.5pt) rectangle (9.5pt,1.5pt);
	\node[anchor=center] at (6.5pt,7pt) {\small $2$};
	\draw [fill=white] (3.5pt,66.5pt) rectangle (9.5pt,56.5pt);
	\node[anchor=center] at (6.5pt,61.5pt) {\small $2$};

\end{tikzpicture}
	\vspace{5pt}
	\caption{Example when a positive swap creates a worse situation for one of the debtors $u_2$.}
	\label{fig:badforu}
\endminipage\hfill
\hspace{0.01\textwidth}
\end{figure}

Furthermore, when our network is not a DAG, a positive swap can even result in a strictly worse shock function for the nodes $u_1$ and $u_2$ that are the debtors of the swapped contracts. This shows that in fact, the debtor nodes $u_1$, $u_2$ might indeed have a reason to object to the swapping operation in this model.

\begin{lemma} \label{lem:badforu}
A positive swap of contracts $(u_1, v_1)$ and $(u_2, v_2)$ can result in a new shock function of $f_{u_2}' < f_{u_2}$ for $u_2$.
\end{lemma}

\begin{proof}
Consider the system in Figure \ref{fig:badforu}, which is again a variation of our motivational example. In the original state of this system, a shock of $k=1$ can remove all the funds of $v_1$ (if it hits $s_1$); on the other hand, it still leaves assets of $\frac{2}{3}$ and $\frac{1}{3}$ at banks $v_2$ and $u_2$, respectively (for these two nodes, the worst case is when the shock hits $s_2$). This means that the original shock functions are $f_{v_1}=(2, 0, 0)$, $f_{v_2}=(2, \frac{2}{3}, 0)$ and $f_{u_2}=(1, \frac{1}{3}, 0)$ for the values $k=(0, 1, 2)$.

Now let us assume that banks $v_1$ and $v_2$ swap their contracts as shown in the dashed arrows. In this new system, the shock function values for $k=2$ and $k=0$ remain the same for each bank. Let us consider the two possible shocks of size $k=1$. If $s_1$ is the one to lose its funds, then $s_2$ can still fulfill its obligations from its own funds, which results in a (direct or indirect) payment of $1$ for all of $v_1$, $v_2$ and $u_2$. On the other hand, if $s_2$ is the one to lose its funds, then $v_1$ receives a payment of $1$ from $u_0$, and it is also a part of the directed cycle in the network; hence the assets $a'_{v_1}$ must satisfy $a'_{v_1} = 1 + \frac{1}{3} \cdot \frac{1}{2} \cdot a'_{v_1}$, giving $a'_{v_1}=\frac{6}{5}$. This means that $s_2$ receives a payment of $\frac{2}{5}$ through the backwards edge, resulting in $a'_{u_2}=\frac{1}{5}$, and $a'_{v_2}=\frac{6}{5}$ for bank $v_2$ which also receives a payment of $1$ from $u_1$.

Since we assume the worst case for each $k$, this results in the shock functions $f'_{v_1}=(2, 1, 0)$, $f'_{v_2}=(2, 1, 0)$ and $f'_{u_2}=(1, \frac{1}{5}, 0)$. Hence we have $f_{v_1}' > f_{v_1}$ and $f_{v_2}' > f_{v_2}$ (so the swap is positive), but $f_{u_2}' < f_{u_2}$.
\end{proof}

\subsection{Complexity analysis}

We now analyze the model from an algorithmic perspective. In this sense, the worst-case models pose a much more serious challenge than the base case of the system or a system with proportional shocks: in particular, it is already computationally hard to find how much a worst-case shock of size $k$ affects a specific bank.

\begin{theorem} \label{th:hardness}
In the worst-set shock model, it is NP-hard to compute $f_v(k)$ for a specific value $k$.
\end{theorem}

One can show this by a reduction from the densest $k$ subgraph problem, which is known to be NP-complete \cite{garey}. Similar reductions have already been shown before for the worst-sum model, and for some other related problems in financial networks \cite{worsttotal}. However, since these earlier results were shown in different model variants, we include a proof of this theorem for completeness.

\begin{proof}
Given an input graph $H$ for the densest $k$ subgraph problem, let us create a separate bank $s$ for each node $z$ of $H$, and select $e_s = \text{deg}(z)$ (the degree of $z$ in $H$). Furthermore, for each edge $(z_1, z_2)$ of $H$, we create another bank $u_{z_1, z_2}$ to represent this edge. Bank $u_{z_1, z_2}$ will have no funds, but we add an incoming debt of $1$ from both of the banks representing $z_1$ and $z_2$. Finally, we add a single sink node $v$ to our financial system (with $e_v=0$), and add a debt of $1$ from each edge node to $v$. Since all banks can fulfill their payment obligations in this system (if no shock happens), the assets of $v$ is equal to the number of edges in $H$ in this network. 

Now let us consider a worst-set shock of size $k$. Since only the source nodes have any funds in this network, any worst shock of size $k$ will hit $k$ source nodes that correspond to $k$ vertices of the original graph $H$. Furthermore, any intermediate bank $u_{z_1, z_2}$ has a simple behavior in this system: if at least one of $z_1$ and $z_2$ is spared by the shock, then $u_{z_1, z_2}$ still provides a payment of $1$ to $v$, but if both $z_1$ and $z_2$ are hit by the shock, then $u_{z_1, z_2}$ makes no payment at all.

This means that for any worst-set shock in this system (i.e. any subset of the nodes selected in $H$), the loss of assets at $v$ is exactly the number of edges $(z_1, z_2)$ in $H$ such that both $z_1$ and $z_2$ are included in the subset, i.e. the number of edges in the corresponding chosen subgraph. As such, finding $f_v(k)$ is equivalent to solving the densest $k$ subgraph problem in the original graph $H$. This completes the reduction.
\end{proof}

Note that this hardness result is for a general $k$ value; for some special cases, e.g. for $k=0$ or $k=n$, the shock function value can easily be computed in polynomial time. In particular, whenever $k$ is a small constant value, the problem also becomes polynomially solvable.

\begin{lemma}
Given a constant value $k$, we can compute $f_v(k)$ in polynomial time.
\end{lemma}

\begin{proof}
The number of possible subsets to be hit by the shock is ${n \choose k} \leq n^k$, which is polynomial in $n$ if $k$ is a constant. We can enumerate all these cases, find the equilibrium (and thus $a_v$) for each case in polynomial time, and simply store the subset that gives the smallest $a_v$ value in order to find $f_v(k)$.
\end{proof}

This also means that the problem becomes polynomially solvable if we only expect a small shock to the network, i.e. we are in the limited worst-set model for some constant limit $K$. In this case, executing the method above for each $k \in \{0, ..., K\}$ still takes polynomial time altogether, and it allows us to describe the entire shock function $f_v$ for this shock model.

These hardness result mean that if we want to analyze a bank's opportunities for swapping in a financial system, then the main computational difficulty in fact lies in this question of evaluating a specific network configuration. If we were given an oracle that returns the shock function of a bank in a specific network, then we could easily evaluate whether a specific swap is positive (or semi-positive), by simply comparing the shock functions in the network before and after the swap. Furthermore, since bank has at most $n$ incoming contracts, such an oracle would also allow us to decide (in polynomial time) whether specific nodes $v_1$, $v_2$ can execute a positive swap in the network, by enumerating each pair of their incoming contracts. As the total number of contracts in the network is also only $O(n^2)$, we could even decide whether there exists a positive swap between \textit{any pair} of acting nodes in the system.

Since the limited worst-set model allows us to compute shock functions efficiently, this also means that these questions are all computationally tractable in this model when $K$ is a small constant.

\begin{corollary} \label{cor:constsmall}
In the limited worst-set model with a constant limit $K$, we can decide in polynomial time whether there is a positive swap in the network.
\end{corollary}

\subsection{Computationally tractable cases}

Since we have seen that deciding whether a swap is positive is NP-hard in the worst-set model, it is a natural question whether the problem is still solvable efficiently in some special class of financial systems. This would imply that we can indeed find positive swap opportunities for some specific financial system structures in practice.

One rather special case of network structures is a \textit{tree network}, i.e. when the undirected version of the debt contacts gives a graph that is a tree.

\begin{theorem} \label{th:treedp}
In a tree network, we can compute $f_v(k)$ for any node $v$ in polynomial time.
\end{theorem}

\begin{proof}
For each specific value $k$, one can compute $f_v(k)$ with a dynamic programming approach, i.e. by computing $f_u(k)$ for each bank $u \in B$ according to an arbitrary topological ordering of the tree, with the help of previously computed values. Initially, for all source nodes $s$ of the tree, $f_s(k)$ can easily be initialized with $f_s(0) = e_s$, and $f_s(k) = 0$ for each $k \geq 1$.

Assume for simplicity first that $u$ is a non-source node in the tree with only two incoming debts from banks $w_1$ and $w_2$, and $e_u=0$. Since the network is a tree, the payments $p_{w_1, u}$ and $p_{w_2, u}$ are independent from each other; as such, the worst shock of size $k$ for $u$ is obtained by splitting the shock into two smaller shocks of size $k_1+k_2=k$, and applying the worst shock of size $k_1$ and $k_2$ on the subtrees rooted at $w_1$ and $w_2$, respectively. Let us briefly introduce $p_{w_i,u}(a_{w_i})$ to denote the payment $p_{w_i, u}$ as a function of $a_{w_i}$, i.e. a simplified notation for the expression $p_{w_i,u}(a_{w_i})=\min(a_{w_i} \! \cdot l_{w_i, u} / l_{w_i},^{\,} l_{w_i, u})$. Then altogether, the shock function of $u$ can be expressed as
\begin{displaymath}
f_u(k) = \min_{k_1 \in [0, k]} \: p_{w_1,u}(^{\,} f_{w_1}(k_1) ^{\,}) + p_{w_2,u}(^{\,} f_{w_2}(k-k_1) ^{\,}) \, .
\end{displaymath}
Since the topological ordering ensures that $f_{w_1}$ and $f_{w_2}$ have already been computed, we can find the value of $f_u$ efficiently for every $k \in [0, n]$.

If $u$ has more than $2$ incoming debts, then we can use this method repeatedly to first distribute the shock between the debtors $w_1$ and $w_2$, then between $\{ w_1, w_2\}$ and $w_3$, and so on, essentially inserting fictitious intermediate banks to limit the indegree to $2$. Since the number of incoming debts is at most $n$, this only adds a linear factor to the computation of $f_u(k)$ at each bank $u$.

Similarly, if we have $e_u>0$, then $u$ can also be hit by the shock; we can cover this case with the previous method by introducing an extra new debtor $s$ which has $e_u$ funds and $l_{s,u}=\infty$.  
\end{proof}

This means that if the network is a tree both before and after the swapping operation, then we can decide in polynomial time whether the given swap is positive (or semi-positive).

Although trees are a very restricted subclass of financial systems, we cannot expect a similar result for a much broader class of networks. Intuitively, whenever the subsets of dependencies of two debtor banks can overlap, this already allows us to express hard combinatorial problems with the network. In fact, the construction in the reduction of Theorem \ref{th:hardness} is a DAG, which shows that the problem is already NP-hard for the slightly more general case when we allow multiple directed paths between two nodes, but not cycles.

Since we have mainly motivated swapping with the idea of removing duplicate dependencies on a specific part of the network, one might wonder if it even makes sense to study swapping in trees. However, it turns out that swapping can still result in an strict improvement for the acting nodes even in trees. Intuitively, while the operation cannot remove multiple dependencies in this case, it can still recombine the parts of the tree into a better configuration which is less vulnerable to shocks of a specific size, resulting in a better shock function in the end. For completeness, we provide such an example in the limited worst-set model.

\begin{theorem} \label{th:treepos}
In the limited worst-set model, there can exist a positive swap in a tree network.
\end{theorem}

\renewcommand{\proofname}{Proof.}

\begin{proof}
The main tool for the proof is a so-called $d$-boolean gadget. Given an integer parameter $d$, this gadget consists of $d$ distinct source nodes $s_1, ..., s_d$ that each have funds of $e_{s_1}\!=...=e_{s_d}\!=d$, and a debt of $l_{s_1, w}\!=...=l_{s_d, w}\!=d$ towards a common intermediate node $w$. This bank $w$ will then have $e_w=0$, and an outgoing debt of weight $d$ to one of our acting nodes. The gadget is useful since it exhibits a threshold behavior: in case of a shock for any $k \leq d-1$, bank $w$ will still make a payment of $d$, but as soon as $k=d$, the payment from $w$ immediately drops to $0$.

Our tree construction consists of a range of such gadgets: $v_1$ has a debt from such gadgets with $d$-values of $3, 4, 5$ and $6$, while $v_2$ has a debt from such gadgets with $d$-values of $3, 4, 6$ and $8$. Consider this network with a shock size limit of $K=10$. Due to the behavior of the gadgets, the loss of $v_1$ for any specific integer $k$ can be computed as the largest subset of $\{3,4,5,6\}$ where the sum is still not larger than $k$. By selecting the appropriate subsets, we will ensure that the subsets received after the swapping are less heterogeneous, which reduces the number of possible sums that can be formed from the subsets.

In particular, for $v_1$, we can create losses of $(0, 0, 0, 3, 4, 5, 6, 7, 8, 9, 10)$ for the values $k=(0, ..., 10)$. Similarly, for $v_2$, the subset $\{3, 4, 6, 8\}$ allows for losses of $(0, 0, 0, 3, 4, 4, 6, 7, 8, 9, 10)$ for $k=(0, ..., 10)$.

Now assume that $v_1$ trades its $4$-boolean gadget for the $3$-boolean gadget of $v_2$. For $v_2$, this provides a new set (well, in fact multiset) of numbers $\{3,3,5,6\}$, allowing for worst-case losses of $(0, 0, 0, 3, 3, 5, 6, 6, 8, 9, 9)$, which is strictly better than before for $k=4, 7$ or $10$. For $v_2$, we get a new multiset of $\{4,4,6,8\}$, producing losses of $(0, 0, 0, 0, 4, 4, 6, 6, 8, 8, 10)$, which is again an improvement for $k=3, 7$ or $9$.

There is one more important technical detail: when $v_1$ exchanges its $4$-boolean gadget for the $3$-boolean gadget of $v_2$, we must also transfer a fixed payment of $1$ from $v_2$ to $v_1$ to keep the total unchanged. That is, if the acting nodes simply swapped two $d$-boolean gadgets with $d=3$ and $d=4$, then this would already decrease the assets of $v_1$ by $1$ in the base case of $k=0$. As such, we need to ensure that $v_1$ also receives an extra asset of $1$ in this swap, and this extra asset has to be ``stable'' in the sense that it is not affected by any shock of size $k \leq K$.

We achieve this by also adding a so-called $1$-fix gadget to our construction. This will be somewhat similar to our $d$-boolean gadgets: given our shock limit $K$, a $1$-fix gadget consists of $K+1$ distinct source nodes that all have funds of $1$, and an outgoing debt of $1$ towards a node $u$. This node $u$ then has a debt of weight $1$ towards some other bank. Since our model assumes that at most $K$ nodes are affected in the shock, this gadget guarantees a payment of $1$ on the outgoing contract of $u$ in any case.

In our construction, we only need to add a single such $1$-fix gadget to the $3$-boolean gadget that $v_2$ is going to swap. More specifically, we add a new intermediate node $w$ that becomes the recipient of the outgoing debt of both the $3$-boolean gadget of $v_2$ and the $1$-fix gadget of $v_2$, and we create an outgoing debt of weight $4$ from this bank $w$ towards $v_2$. As such, this bank $w$ has a shock function of $(4, 4, 4, 1, 1, ..., 1)$ for $k=(0, 1, 2, 3, 4, ..., K)$, and $v_2$ can only swap this part of the tree if it redirects the contract from $w$ towards $v_1$, which gives both the payments from the $3$-boolean gadget and the $1$-fix gadget to $v_1$. Note that the network still remains a tree after this technical step.

Given this extended network, the total assets of the acting nodes (in case of no shock) becomes the same before and after the swap: we have $a_{v_1}=a_{v_1}'=18$ and $a_{v_2}=a_{v_2}'=22$. Given the subset of numbers $\{3,4,5,6\}$ for $v_1$, the shock function of $v_1$ becomes $(18, 18, 18, 15, 14, 13, 12, 11, 10, 9, 8)$ for the parameters $k=(0, ..., 10)$ before the swap. Similarly, with the values $\{3, 4, 6, 8\}$, the shock function of $v_2$ is $(22, 22, 22, 19, 18, 18, 16, 15, 14, 13, 12)$ for $k=(0, ..., 10)$ initially.

Now assume that $v_1$ trades a $4$-boolean gadget for the $3$-boolean gadget of $v_2$ which also has the $1$-fix gadget attached. Then the new subset of $v_1$ is $\{3,3,5,6\}$, giving a shock function of $(18, 18, 18, 15, 15, 13, 12, 12, 10, 9, 9)$ for $k=(0, ..., 10)$, which is indeed an improvement for the values $k=4, 7$ and $10$. Similarly, the new subset of $v_2$ is $\{4,4,6,8\}$, producing a shock function of $(22, 22, 22, 22, 18, 18, 16, 16, 14, 14, 12)$ for $k=(0, ..., 10)$; again a strict improvement for $k=3, 7$ and $9$.
\end{proof}

\renewcommand{\proofname}{Proof.}

\subsection{An optimization perspective}

Finally, we look at the swapping operation from an optimization perspective, i.e. as an improvement step in the search space of acting banks trying to find a better network configuration. In this section, we show that swapping is a local operation in the sense that it can easily get stuck in a local minimum. In particular, there are cases where the situation of banks cannot be improved with a simple swap, but it can be improved with more sophisticated operations like portfolio swapping or debt reorganization. 

\begin{theorem} \label{th:portfolio}
It is possible that for two acting nodes $v_1$ and $v_2$, there is not even a semi-positive swap in the network, but there is a positive portfolio swap.
\end{theorem}

\begin{proof}
Consider two source nodes $s_1$ and $s_2$ that both have funds of $e_{s_1}=e_{s_2}=76$. Then consider a second layer of $2 \times 4$ intermediate nodes that have the following pairs of debts from $s_1$ and $s_2$, respectively: $(2,17)$, $(7,12)$, $(12,7)$, $(16,3)$ for the first $4$ nodes, and $(5,14)$, $(5,14)$, $(10,9)$, $(19,0)$ for the second $4$ nodes. Note that $s_1$ and $s_2$ can fulfill all of these obligations unless they are hit by the shock; as such, the assets of e.g. the first intermediate node can be expressed as $2 \cdot r_{s_1} + 17 \cdot r_{s_2}$.

Finally, let acting node $v_1$ have incoming debts of weight $19$ from the first $4$ intermediate nodes, and $v_2$ have incoming debts of weight $19$ from the second $4$ intermediate nodes. In this system, the assets of $v_1$ can be computed as $a_{v_1}=37 \cdot r_{s_1} + 39 \cdot r_{s_2}$, while the assets of $v_2$ are $a_{v_2}= 39 \cdot r_{s_1} + 37 \cdot r_{s_2}$. This means that both $v_1$ and $v_2$ have a shock function of $(76, 37, 0)$ for the values $k=(0, 1, 2)$, respectively.

One can check that after swapping any single pair of contracts, the situation of $v_1$ and $v_2$ gets worse: for both banks, the coefficient of dependency on either $r_{s_1}$ or $r_{s_2}$ becomes strictly larger than 39, so their remaining assets in the $k=1$ case fall strictly below $37$.

However, if $v_1$ swaps its contracts $(7, 12)$ and $(16,3)$ for the contracts $(5, 14)$ and $(19, 0)$ owned by $v_2$, then both of them will have $a_{v_1}'=a_{v_2}'=38 \cdot r_{s_1} + 38 \cdot r_{s_2}$. This results in a shock function of $(76, 38, 0)$ for $k=(0, 1, 2)$, thus indeed providing an improvement for the $k=1$ case.
\end{proof}

\begin{theorem} \label{th:reorg}
Given a triplet of nodes $v_1$, $v_2$, $v_3$, it is possible that there is no positive swap for any pair of these banks, but there is a positive debt reorganization for the $3$ banks together.
\end{theorem}

\begin{proof}
Consider $3$ source nodes $s_1, s_2, s_3$ with $e_{s_1}=e_{s_2}=e_{s_3} = 3$. For each of them, we add 3 intermediate nodes with a debt of $1$ from the respective source; e.g. for $s_1$, this gives us $3$ banks $u_{1,1}, u_{1,2}, u_{1,3}$ that receive $a_{u_{1,1}}=a_{u_{1,2}}=a_{u_{1,3}}=r_{s_1}$ from the system.

We then ensure that each acting node has a unit-weight debt from $3$ intermediate nodes: $v_1$ from $u_{1,1}$, $u_{1,2}$ and $u_{2,3}$; $v_2$ from $u_{2,1}$, $u_{2,2}$ and $u_{3,3}$; and $v_3$ from $u_{3,1}$, $u_{3,2}$ and $u_{1,3}$. This way, the acting nodes express the following asset functions: $a_{v_1} = 2 \cdot r_{s_1} + r_{s_2}$, $a_{v_2} = 2 \cdot r_{s_2} + r_{s_3}$, $a_{v_3} = 2 \cdot r_{s_3} + r_{s_1}$. As such, the shock function of the banks is $(3, 1, 0, 0)$ for $k=(0, 1, 2, 3)$, respectively.

One can check that no two banks have a positive swap for any single pair of contracts. For example, $v_1$ and $v_2$ have incoming payments of $r_{s_1}, r_{s_1}, r_{s_2}$ and $r_{s_2}, r_{s_2}, r_{s_3}$, respectively. The shock function of $v_1$ is only improved if it trades an $r_{s_1}$ for getting an $r_{s_3}$ from $v_2$; however, then $v_2$ ends up with $r_{s_2}, r_{s_2}, r_{s_1}$, having the same shock function $(3, 1, 0, 0)$ as in the original case.

However, one can easily reorganize the whole system such that each node has an intermediate debtor from each source (e.g. $v_1$ has debts from $u_{1,1}$, $u_{2,2}$ and $u_{3,3}$), resulting in $a_{v_1}'=a_{v_2}'=a_{v_3}'= r_{s_1} + r_{s_2} + r_{s_3}$. This provides a shock function of $(3, 2, 1, 0)$ for $k=(0, 1, 2, 3)$, which is indeed an improvement.
\end{proof}

\section{Conclusion}

In this paper we have studied whether banks can gain more assets or mitigate the effects of external shocks by executing a debt swap in a financial network. Swapping is a very simple local operation that keeps the incoming and outgoing liabilities unchanged, so it is a natural choice for a fundamental reorganization step in the network. Our results show that the base model and the proportional shock model do not allow for a swap that is beneficial for both of the acting nodes; on the other hand, a positive swap is often possible in models where the goal of banks is to mitigate their losses in the worst possible case. However, these models also raise some difficult questions: other banks in the network might be affected by the swap negatively, and it also becomes hard to analyze the system from a computational perspective.

We point out that while our result for the base and proportional models is mostly negative, there are various ways to extend these models by further practical aspects that make a positive swap possible. We have already noted at the end of Section \ref{sec:nopos} that one such case is when we also account for the administrative costs of a default. In another approach, the acting banks could also take advantage of semi-positive swaps where $a_{v_1}'>a_{v_1}$ and $a_{v_2}'=a_{v_2}$: if $v_1$ pays a small fee to $v_2$ (with or without the knowledge of regulators) in order to motivate $v_2$ to agree to the trade, then the transaction suddenly becomes beneficial for both parties. Similarly, banks in the proportional model could have confident but different expectations about the future of the market: if $v_1$ wants to maximize its assets for a shock of size $\lambda_1$, and $v_2$ wants to maximize its assets for a shock of some other size $\lambda_2$, then there can again be a swap that is beneficial for both of them. We leave it to future work to study different network operations in these more general settings.

\bibliographystyle{ACM-Reference-Format}
\bibliography{references}


\begin{thebibliography}{26}


\ifx \showCODEN    \undefined \def \showCODEN     #1{\unskip}     \fi
\ifx \showDOI      \undefined \def \showDOI       #1{#1}\fi
\ifx \showISBNx    \undefined \def \showISBNx     #1{\unskip}     \fi
\ifx \showISBNxiii \undefined \def \showISBNxiii  #1{\unskip}     \fi
\ifx \showISSN     \undefined \def \showISSN      #1{\unskip}     \fi
\ifx \showLCCN     \undefined \def \showLCCN      #1{\unskip}     \fi
\ifx \shownote     \undefined \def \shownote      #1{#1}          \fi
\ifx \showarticletitle \undefined \def \showarticletitle #1{#1}   \fi
\ifx \showURL      \undefined \def \showURL       {\relax}        \fi
\providecommand\bibfield[2]{#2}
\providecommand\bibinfo[2]{#2}
\providecommand\natexlab[1]{#1}
\providecommand\showeprint[2][]{arXiv:#2}

\bibitem[\protect\citeauthoryear{Acemoglu, Ozdaglar, and
  Tahbaz-Salehi}{Acemoglu et~al\mbox{.}}{2015}]%
        {gen2}
\bibfield{author}{\bibinfo{person}{Daron Acemoglu}, \bibinfo{person}{Asuman
  Ozdaglar}, {and} \bibinfo{person}{Alireza Tahbaz-Salehi}.}
  \bibinfo{year}{2015}\natexlab{}.
\newblock \showarticletitle{Systemic risk and stability in financial networks}.
\newblock \bibinfo{journal}{\emph{American Economic Review}}
  \bibinfo{volume}{105}, \bibinfo{number}{2} (\bibinfo{year}{2015}),
  \bibinfo{pages}{564--608}.
\newblock


\bibitem[\protect\citeauthoryear{Amini, Filipovi{\'c}, and Minca}{Amini
  et~al\mbox{.}}{2016}]%
        {CCP4}
\bibfield{author}{\bibinfo{person}{Hamed Amini}, \bibinfo{person}{Damir
  Filipovi{\'c}}, {and} \bibinfo{person}{Andreea Minca}.}
  \bibinfo{year}{2016}\natexlab{}.
\newblock \showarticletitle{To fully net or not to net: Adverse effects of
  partial multilateral netting}.
\newblock \bibinfo{journal}{\emph{Operations Research}} \bibinfo{volume}{64},
  \bibinfo{number}{5} (\bibinfo{year}{2016}), \bibinfo{pages}{1135--1142}.
\newblock


\bibitem[\protect\citeauthoryear{Bertschinger, Hoefer, and
  Schmand}{Bertschinger et~al\mbox{.}}{2020}]%
        {gametheo}
\bibfield{author}{\bibinfo{person}{Nils Bertschinger}, \bibinfo{person}{Martin
  Hoefer}, {and} \bibinfo{person}{Daniel Schmand}.}
  \bibinfo{year}{2020}\natexlab{}.
\newblock \showarticletitle{{Strategic Payments in Financial Networks}}. In
  \bibinfo{booktitle}{\emph{11th Innovations in Theoretical Computer Science
  Conference (ITCS 2020)}} \emph{(\bibinfo{series}{LIPIcs},
  Vol.~\bibinfo{volume}{151})}. \bibinfo{publisher}{Schloss
  Dagstuhl--Leibniz-Zentrum f{\"u}r Informatik}, \bibinfo{address}{Dagstuhl,
  Germany}, \bibinfo{pages}{46:1--46:16}.
\newblock
\showISBNx{978-3-95977-134-4}
\showISSN{1868-8969}


\bibitem[\protect\citeauthoryear{Cs{\'o}ka and Jean-Jacques~Herings}{Cs{\'o}ka
  and Jean-Jacques~Herings}{2017}]%
        {seqclear1}
\bibfield{author}{\bibinfo{person}{P{\'e}ter Cs{\'o}ka} {and}
  \bibinfo{person}{P Jean-Jacques~Herings}.} \bibinfo{year}{2017}\natexlab{}.
\newblock \showarticletitle{Decentralized clearing in financial networks}.
\newblock \bibinfo{journal}{\emph{Management Science}} \bibinfo{volume}{64},
  \bibinfo{number}{10} (\bibinfo{year}{2017}), \bibinfo{pages}{4681--4699}.
\newblock


\bibitem[\protect\citeauthoryear{Cui, Feng, Hu, and Zou}{Cui
  et~al\mbox{.}}{2018}]%
        {CCP2}
\bibfield{author}{\bibinfo{person}{Zhenyu Cui}, \bibinfo{person}{Qi Feng},
  \bibinfo{person}{Ruimeng Hu}, {and} \bibinfo{person}{Bin Zou}.}
  \bibinfo{year}{2018}\natexlab{}.
\newblock \showarticletitle{Systemic risk and optimal fee for central clearing
  counterparty under partial netting}.
\newblock \bibinfo{journal}{\emph{Operations Research Letters}}
  \bibinfo{volume}{46}, \bibinfo{number}{3} (\bibinfo{year}{2018}),
  \bibinfo{pages}{306--311}.
\newblock


\bibitem[\protect\citeauthoryear{Dees, Henry, and Martin}{Dees
  et~al\mbox{.}}{2017}]%
        {ecb}
\bibfield{author}{\bibinfo{person}{St{\'e}phane Dees},
  \bibinfo{person}{J{\'e}r{\^o}me Henry}, {and} \bibinfo{person}{Reiner
  Martin}.} \bibinfo{year}{2017}\natexlab{}.
\newblock \showarticletitle{STAMP\euro{}: stress-test analytics for
  macroprudential purposes in the euro area}.
\newblock \bibinfo{journal}{\emph{Frankfurt am Main: ECB}}
  (\bibinfo{year}{2017}).
\newblock


\bibitem[\protect\citeauthoryear{Demange}{Demange}{2016}]%
        {gen3}
\bibfield{author}{\bibinfo{person}{Gabrielle Demange}.}
  \bibinfo{year}{2016}\natexlab{}.
\newblock \showarticletitle{Contagion in financial networks: a threat index}.
\newblock \bibinfo{journal}{\emph{Management Science}} \bibinfo{volume}{64},
  \bibinfo{number}{2} (\bibinfo{year}{2016}), \bibinfo{pages}{955--970}.
\newblock


\bibitem[\protect\citeauthoryear{Duffie, Scheicher, and Vuillemey}{Duffie
  et~al\mbox{.}}{2015}]%
        {CCP1}
\bibfield{author}{\bibinfo{person}{Darrell Duffie}, \bibinfo{person}{Martin
  Scheicher}, {and} \bibinfo{person}{Guillaume Vuillemey}.}
  \bibinfo{year}{2015}\natexlab{}.
\newblock \showarticletitle{Central clearing and collateral demand}.
\newblock \bibinfo{journal}{\emph{Journal of Financial Economics}}
  \bibinfo{volume}{116}, \bibinfo{number}{2} (\bibinfo{year}{2015}),
  \bibinfo{pages}{237--256}.
\newblock


\bibitem[\protect\citeauthoryear{Duffie and Zhu}{Duffie and Zhu}{2011}]%
        {CCP3}
\bibfield{author}{\bibinfo{person}{Darrell Duffie} {and}
  \bibinfo{person}{Haoxiang Zhu}.} \bibinfo{year}{2011}\natexlab{}.
\newblock \showarticletitle{Does a central clearing counterparty reduce
  counterparty risk?}
\newblock \bibinfo{journal}{\emph{The Review of Asset Pricing Studies}}
  \bibinfo{volume}{1}, \bibinfo{number}{1} (\bibinfo{year}{2011}),
  \bibinfo{pages}{74--95}.
\newblock


\bibitem[\protect\citeauthoryear{Eisenberg and Noe}{Eisenberg and Noe}{2001}]%
        {model1}
\bibfield{author}{\bibinfo{person}{Larry Eisenberg} {and}
  \bibinfo{person}{Thomas~H Noe}.} \bibinfo{year}{2001}\natexlab{}.
\newblock \showarticletitle{Systemic risk in financial systems}.
\newblock \bibinfo{journal}{\emph{Management Science}} \bibinfo{volume}{47},
  \bibinfo{number}{2} (\bibinfo{year}{2001}), \bibinfo{pages}{236--249}.
\newblock


\bibitem[\protect\citeauthoryear{Elliott, Golub, and Jackson}{Elliott
  et~al\mbox{.}}{2014}]%
        {cross2}
\bibfield{author}{\bibinfo{person}{Matthew Elliott}, \bibinfo{person}{Benjamin
  Golub}, {and} \bibinfo{person}{Matthew~O Jackson}.}
  \bibinfo{year}{2014}\natexlab{}.
\newblock \showarticletitle{Financial networks and contagion}.
\newblock \bibinfo{journal}{\emph{American Economic Review}}
  \bibinfo{volume}{104}, \bibinfo{number}{10} (\bibinfo{year}{2014}),
  \bibinfo{pages}{3115--53}.
\newblock


\bibitem[\protect\citeauthoryear{Feinstein, Pang, Rudloff, Schaanning, Sturm,
  and Wildman}{Feinstein et~al\mbox{.}}{2018}]%
        {sensitivity}
\bibfield{author}{\bibinfo{person}{Zachary Feinstein}, \bibinfo{person}{Weijie
  Pang}, \bibinfo{person}{Birgit Rudloff}, \bibinfo{person}{Eric Schaanning},
  \bibinfo{person}{Stephan Sturm}, {and} \bibinfo{person}{Mackenzie Wildman}.}
  \bibinfo{year}{2018}\natexlab{}.
\newblock \showarticletitle{Sensitivity of the Eisenberg--Noe Clearing Vector
  to Individual Interbank Liabilities}.
\newblock \bibinfo{journal}{\emph{SIAM Journal on Financial Mathematics}}
  \bibinfo{volume}{9}, \bibinfo{number}{4} (\bibinfo{year}{2018}),
  \bibinfo{pages}{1286--1325}.
\newblock


\bibitem[\protect\citeauthoryear{Garey and Johnson}{Garey and Johnson}{1979}]%
        {garey}
\bibfield{author}{\bibinfo{person}{Michael~R. Garey} {and}
  \bibinfo{person}{David~S. Johnson}.} \bibinfo{year}{1979}\natexlab{}.
\newblock \bibinfo{booktitle}{\emph{Computers and Intractability: A Guide to
  the Theory of NP-Completeness}}.
\newblock \bibinfo{publisher}{W. H. Freeman \& Co.}
\newblock
\showISBNx{0716710447}


\bibitem[\protect\citeauthoryear{Gavrila and Popa}{Gavrila and Popa}{2020}]%
        {cycles3}
\bibfield{author}{\bibinfo{person}{Lucian-Ionut Gavrila} {and}
  \bibinfo{person}{Alexandru Popa}.} \bibinfo{year}{2020}\natexlab{}.
\newblock \bibinfo{title}{A novel algorithm for clearing financial obligations
  between companies -- an application within the Romanian Ministry of Economy}.
\newblock
\newblock
\showeprint[arxiv]{2012.05564}~[cs.DS]


\bibitem[\protect\citeauthoryear{Glasserman and Young}{Glasserman and
  Young}{2015}]%
        {gen1}
\bibfield{author}{\bibinfo{person}{Paul Glasserman} {and}
  \bibinfo{person}{H~Peyton Young}.} \bibinfo{year}{2015}\natexlab{}.
\newblock \showarticletitle{How likely is contagion in financial networks?}
\newblock \bibinfo{journal}{\emph{Journal of Banking \& Finance}}
  \bibinfo{volume}{50} (\bibinfo{year}{2015}), \bibinfo{pages}{383--399}.
\newblock


\bibitem[\protect\citeauthoryear{Hemenway and Khanna}{Hemenway and
  Khanna}{2016}]%
        {worsttotal}
\bibfield{author}{\bibinfo{person}{Brett Hemenway} {and}
  \bibinfo{person}{Sanjeev Khanna}.} \bibinfo{year}{2016}\natexlab{}.
\newblock \showarticletitle{Sensitivity and computational complexity in
  financial networks}.
\newblock \bibinfo{journal}{\emph{Algorithmic Finance}} \bibinfo{volume}{5},
  \bibinfo{number}{3-4} (\bibinfo{year}{2016}), \bibinfo{pages}{95--110}.
\newblock


\bibitem[\protect\citeauthoryear{Leduc, Poledna, and Thurner}{Leduc
  et~al\mbox{.}}{2017}]%
        {coveredCDS}
\bibfield{author}{\bibinfo{person}{Matt~V Leduc}, \bibinfo{person}{Sebastian
  Poledna}, {and} \bibinfo{person}{Stefan Thurner}.}
  \bibinfo{year}{2017}\natexlab{}.
\newblock \showarticletitle{Systemic risk management in financial networks with
  credit default swaps}.
\newblock \bibinfo{journal}{\emph{Available at SSRN 2713200}}
  (\bibinfo{year}{2017}).
\newblock


\bibitem[\protect\citeauthoryear{Papp and Wattenhofer}{Papp and
  Wattenhofer}{2020a}]%
        {arxiv}
\bibfield{author}{\bibinfo{person}{Pál~András Papp} {and}
  \bibinfo{person}{Roger Wattenhofer}.} \bibinfo{year}{2020}\natexlab{a}.
\newblock \bibinfo{title}{Default Ambiguity: Finding the Best Solution to the
  Clearing Problem}.
\newblock
\newblock
\newblock
\shownote{ArXiv preprint arXiv:2002.07741.}


\bibitem[\protect\citeauthoryear{Papp and Wattenhofer}{Papp and
  Wattenhofer}{2020b}]%
        {icalp}
\bibfield{author}{\bibinfo{person}{P{\'a}l~Andr{\'a}s Papp} {and}
  \bibinfo{person}{Roger Wattenhofer}.} \bibinfo{year}{2020}\natexlab{b}.
\newblock \showarticletitle{{Network-Aware Strategies in Financial Systems}}.
  In \bibinfo{booktitle}{\emph{47th International Colloquium on Automata,
  Languages, and Programming (ICALP 2020)}} \emph{(\bibinfo{series}{LIPIcs},
  Vol.~\bibinfo{volume}{168})}. \bibinfo{publisher}{Schloss
  Dagstuhl--Leibniz-Zentrum f{\"u}r Informatik}, \bibinfo{address}{Dagstuhl,
  Germany}, \bibinfo{pages}{91:1--91:17}.
\newblock
\showISBNx{978-3-95977-138-2}
\showISSN{1868-8969}


\bibitem[\protect\citeauthoryear{Papp and Wattenhofer}{Papp and
  Wattenhofer}{2021}]%
        {itcs}
\bibfield{author}{\bibinfo{person}{P{\'a}l~Andr{\'a}s Papp} {and}
  \bibinfo{person}{Roger Wattenhofer}.} \bibinfo{year}{2021}\natexlab{}.
\newblock \showarticletitle{{Sequential Defaulting in Financial Networks}}. In
  \bibinfo{booktitle}{\emph{12th Innovations in Theoretical Computer Science
  Conference (ITCS 2021)}} \emph{(\bibinfo{series}{LIPIcs},
  Vol.~\bibinfo{volume}{185})}. \bibinfo{publisher}{Schloss
  Dagstuhl--Leibniz-Zentrum f{\"u}r Informatik}, \bibinfo{address}{Dagstuhl,
  Germany}, \bibinfo{pages}{52:1--52:20}.
\newblock


\bibitem[\protect\citeauthoryear{Rogers and Veraart}{Rogers and
  Veraart}{2013}]%
        {veraart}
\bibfield{author}{\bibinfo{person}{Leonard~CG Rogers} {and}
  \bibinfo{person}{Luitgard~AM Veraart}.} \bibinfo{year}{2013}\natexlab{}.
\newblock \showarticletitle{Failure and rescue in an interbank network}.
\newblock \bibinfo{journal}{\emph{Management Science}} \bibinfo{volume}{59},
  \bibinfo{number}{4} (\bibinfo{year}{2013}), \bibinfo{pages}{882--898}.
\newblock


\bibitem[\protect\citeauthoryear{Schuldenzucker and Seuken}{Schuldenzucker and
  Seuken}{2020}]%
        {cycles}
\bibfield{author}{\bibinfo{person}{Steffen Schuldenzucker} {and}
  \bibinfo{person}{Sven Seuken}.} \bibinfo{year}{2020}\natexlab{}.
\newblock \showarticletitle{Portfolio Compression in Financial Networks:
  Incentives and Systemic Risk}. In \bibinfo{booktitle}{\emph{Proceedings of
  the 21st ACM Conference on Economics and Computation}} (Virtual Event,
  Hungary) \emph{(\bibinfo{series}{EC '20})}. \bibinfo{publisher}{Association
  for Computing Machinery}, \bibinfo{address}{New York, NY, USA},
  \bibinfo{pages}{79}.
\newblock


\bibitem[\protect\citeauthoryear{Schuldenzucker, Seuken, and
  Battiston}{Schuldenzucker et~al\mbox{.}}{2016}]%
        {cds1}
\bibfield{author}{\bibinfo{person}{Steffen Schuldenzucker},
  \bibinfo{person}{Sven Seuken}, {and} \bibinfo{person}{Stefano Battiston}.}
  \bibinfo{year}{2016}\natexlab{}.
\newblock \showarticletitle{Clearing Payments in Financial Networks with Credit
  Default Swaps}. In \bibinfo{booktitle}{\emph{Proceedings of the 2016 ACM
  Conference on Economics and Computation}} (Maastricht, The Netherlands)
  \emph{(\bibinfo{series}{EC '16})}. \bibinfo{publisher}{ACM},
  \bibinfo{address}{New York, NY, USA}, \bibinfo{pages}{759--759}.
\newblock
\showISBNx{978-1-4503-3936-0}


\bibitem[\protect\citeauthoryear{Schuldenzucker, Seuken, and
  Battiston}{Schuldenzucker et~al\mbox{.}}{2017}]%
        {cds2}
\bibfield{author}{\bibinfo{person}{Steffen Schuldenzucker},
  \bibinfo{person}{Sven Seuken}, {and} \bibinfo{person}{Stefano Battiston}.}
  \bibinfo{year}{2017}\natexlab{}.
\newblock \showarticletitle{{Finding Clearing Payments in Financial Networks
  with Credit Default Swaps is PPAD-complete}}. In
  \bibinfo{booktitle}{\emph{8th Innovations in Theoretical Computer Science
  Conference (ITCS 2017)}} \emph{(\bibinfo{series}{LIPIcs},
  Vol.~\bibinfo{volume}{67})}. \bibinfo{publisher}{Schloss
  Dagstuhl--Leibniz-Zentrum f{\"u}r Informatik}, \bibinfo{address}{Dagstuhl,
  Germany}, \bibinfo{pages}{32:1--32:20}.
\newblock
\showISBNx{978-3-95977-029-3}
\showISSN{1868-8969}


\bibitem[\protect\citeauthoryear{Veraart}{Veraart}{2020}]%
        {cycles2}
\bibfield{author}{\bibinfo{person}{Luitgard Anna~Maria Veraart}.}
  \bibinfo{year}{2020}\natexlab{}.
\newblock \showarticletitle{When does portfolio compression reduce systemic
  risk?}
\newblock \bibinfo{journal}{\emph{Available at SSRN 3688495}}
  (\bibinfo{year}{2020}).
\newblock


\bibitem[\protect\citeauthoryear{Vitali, Glattfelder, and Battiston}{Vitali
  et~al\mbox{.}}{2011}]%
        {cross1}
\bibfield{author}{\bibinfo{person}{Stefania Vitali}, \bibinfo{person}{James~B
  Glattfelder}, {and} \bibinfo{person}{Stefano Battiston}.}
  \bibinfo{year}{2011}\natexlab{}.
\newblock \showarticletitle{The network of global corporate control}.
\newblock \bibinfo{journal}{\emph{PloS one}} \bibinfo{volume}{6},
  \bibinfo{number}{10} (\bibinfo{year}{2011}), \bibinfo{pages}{e25995}.
\newblock


\end{thebibliography}

\appendix

\section{Formal definitions and basic properties} \label{App:A}

For completeness, we now give a more formal definition of the more sophisticated swapping operations.

\begin{definition}[Portfolio Swapping]
Assume $G$ has two nodes $v_1$, $v_2$ and two disjoint node sets $U_1$, $U_2$ that do not contain $v_1$ or $v_2$. Assume that $\sum_{u_i \in U_1} l_{u_{i}, v_1}=\sum_{u_i \in U_2} l_{u_{i}, v_2}$, and also $l_{u_i, v_2}=0$ for all $u_i \in U_1$, $l_{u_i, v_1}=0$ for all $u_i \in U_2$.

A \emph{portfolio swapping} produces a new network $G'$ that only differs from $G$ in the following:
\begin{itemize}
\setlength{\itemsep}{1pt}
\setlength{\topsep}{1pt}
 \item for all $u_i \in U_1$, we have $l_{u_i, v_2}'=l_{u_i, v_1}$ and $l_{u_i, v_1}'=0$,
 \item for all $u_i \in U_2$, we have $l_{u_i, v_1}'=l_{u_i, v_2}$ and $l_{u_i, v_2}'=0$.
\end{itemize}
\end{definition}

\begin{definition}[Debt Reorganization]
Assume $G$ has a set of nodes $u_1$, $u_2$, ..., $u_m$ and $v_1$, $v_2$, ..., $v_m$ such that there is a specific value $d$ such that $l_{u_i, v_i}=d$ for all $i \in \{ 1, ..., m \}$, and $l_{u_i, v_j}=0$ for all combinations of $i,j \in \{ 1, ..., m \}$, $i \neq j$.

Given a permutation $\Pi: \{ 1, ..., m \} \rightarrow \{ 1, ..., m \}$ of the numbers $1,...,m$, a \emph{debt reorganization} according to this permutation produces a new network $G'$ that only differs from $G$ in the following:
\begin{itemize}
\setlength{\itemsep}{1pt}
\setlength{\topsep}{1pt}
 \item for all $i \in \{ 1, ..., m \}$, we have $l_{u_i, v_{\Pi(i)}}'=d$,
 \item for all $i, j \in \{ 1, ..., m \}$ with $j \neq \Pi(i)$, we have $l_{u_i, v_j}'=0$.
\end{itemize}
We usually assume that the underlying permutation $\Pi$ is fixed-point free, i.e. $\Pi(i) \neq i$ for all $i \in \{ 1, ..., m \}$.
\end{definition}

We continue by discussing the basic properties of financial systems that are used throughout Section \ref{sec:base}. We once again point out that some of these properties have already been discussed (in our network model or a related one) in previous works, e.g. in \cite{model1} or \cite{veraart}.

\paragraph*{Monotonicity.} One of the most fundamental properties of our networks is monotonicity, or in other words, the fact that banks have \textit{long positions} on each other. Intuitively, this means that if a specific node is in a better situation, then this can only result in a better (or the same) situation for other nodes, but never a worse one. Lemma \ref{lem:monot} formulates this property for the special case of a source $s$ and a sink $t$, but it also holds for any pair of banks in the system.

One can prove this property by defining a new financial network $G_{\text{diff}}$ that corresponds to the ``difference'' of the two systems (before and after the increase). Given the same set of nodes and contracts, we can consider the remaining liability $l_{u,v}\,\!^{(G_{\text{diff}})} = l_{u,v} - p_{u,v}$ on each contract of this difference system (with $l_{u,v}$ and $p_{u,v}$ understood on the edge in the original state of the system). For convenience, contracts with $l_{u,v}\,\!^{(G_{\text{diff}})} = 0$ can be dropped. We then assign $e_u\,\!^{(G_{\text{diff}})} = 0$ to each bank, apart from the source $s$, where we choose $e_s\,\!^{(G_{\text{diff}})}=\hat{e}_s-e_s$. 

Since this system describes the remaining payment obligations that are unsatisfied in the original system, the solution of the increased system with $\hat{e}_s$ is obtained as the sum of the solution of the original system, plus the solution of this difference system.  More formally, one can observe that any such sum provides a clearing vector in the increased system, and any clearing vector in the increased system defines a clearing vector in the difference system, so the extra funds of $s$ will indeed be distributed according to the solution of the difference system. Throughout the proofs of Section \ref{sec:base}, we often analyze such difference systems to study the result of increasing the funds of a specific bank.

Monotonicity already follows from this observation, since we have $\hat{a}_u = a_u + a_u\,\!^{(G_{\text{diff}})}$ where $a_u\,\!^{(G_{\text{diff}})} \geq 0$. We can also generalize this monotonicity property for the case when the funds of multiple banks are increased, simply by executing these increases one after the other.

\paragraph*{Indirect monotonicity.}
As a technical detail, we note that we also use a slightly different version of monotonicity when we have two sources $s_1$ and $s_2$, and the funds are changed to $\hat{e}_{s_1} < e_{s_1}$ and $\hat{e}_{s_2} > e_{s_2}$. In this case, the funds of a source $s_1$ are decreased; however, if we know for a fact that for every debtor $u$ of $s_1$ we have $\hat{a}_u \geq a_u$ (i.e. their losses from $s_1$ are compensated by the increase of funds at $s_2$), then the decrease at $s_1$ can have no negative impact on the rest of the network: $\hat{a}_w \geq a_w$ still holds for each bank $w \neq s_1$.

Formally, the proof of this property is as follows. For simplicity, we focus on the case when $s_1$ only has a single debtor $u$. If $u$ has no liabilities at all, then the decrease can only affect $u$, so the assumption of $\hat{a}_u \geq a_u$ already settles the claim. Otherwise, we can slightly modify the difference system by (i) increasing the funds of $s_1$ back to $e_{s_1}$, and (ii) adding a new sink $\hat{t}$ and a debt from $u$ to $\hat{t}$ such that the payment on this debt is exactly $e_{s_1} - \hat{e}_{s_1}$ (one can easily compute the exact liability $l_{u, \hat{t}}$ required for this). This ensures that the incoming and outgoing payments of $u$ are both increased by $e_{s_1} - \hat{e}_{s_1}$, so the payment remains unchanged on every other outgoing debt of $u$. As such, every other bank $v$ will have the same assets in this new difference system as before. This modified system also provides non-negative $a_v\,\!^{(G_{\text{diff}})} \geq 0$, which again implies $\hat{a}_v \geq a_v$.

\paragraph*{Non-expansivity.}
The non-expansive property of Lemma \ref{lem:nonexp} follows from the fact that in the difference system corresponding to an increase of $\Delta$, the source node $s$ only has $\Delta$ new assets, every other node has at least as much incoming as outgoing payment, and there can be no outgoing payment from $t$.

More formally, if we use the notation $p_{u, \text{in}} = \sum_{v \in B} \, p_{v, u}$ and $p_{u, \text{out}} = \sum_{v \in B} \, p_{u,v}$ in the difference system, then $\sum_{u \in B} \, p_{u, \text{in}} = \sum_{u \in B} \, p_{u, \text{out}}$, since we are simply counting the same contracts in two different ways. Every intermediate node (apart from $s$) has $p_{u, \text{in}} \geq p_{u, \text{out}}$ because none of them have any funds in the difference system; as such, we can subtract these inequalities to obtain $p_{s, \text{in}} + p_{t, \text{in}} \leq p_{s, \text{out}} + p_{t, \text{out}}$. Since $s$ and $t$ are source and sink nodes, respectively, this simplifies to $p_{t, \text{in}} \leq p_{s, \text{out}}$. Since we know that $p_{s, \text{out}} \leq e_s = \Delta$, the property follows.

\paragraph*{Linearity.}
We have already discussed that linearity is the special case when all the extra units of funds injected at $s$ will end up as an asset of $t$ after traversing some route through the network. This also implies that all the new assets received by any intermediate node $u$ in this increase will be relayed in the system to some other node (a creditor of $u$), with the exception of the sink $t$. This in turn also means that all such intermediate banks $u$ are in default, and there are still unpaid liabilities on every contract that is contained in a directed path starting from $s$.

We point out that given a set $T$ of sink nodes, our proofs often use the fact that a bank $s_1$ is $T$-linear on some interval $[x_1, x_1 + \Delta_1]$ and another bank $s_2$ is $T$-linear on some other interval $[x_2, x_2 + \Delta_2]$. However, formally, this is not precise, as the $T$-linearity of $s_1$ also depends on the current value of $e_{s_2}$, and vice versa: a higher amount of funds at $s_2$ can result in higher payments in the network, which might mean that there are no more remaining liabilities on some of the directed paths from $s_1$ to $T$. Our observations in the proofs are always about the fact that $s_1$ is $T$-linear on $[x_1, x_1 + \Delta_1]$ assuming that $e_{s_2}=x_2$, and $s_2$ is $T$-linear on $[x_2, x_2 + \Delta_2]$ assuming that $e_{s_1}=x_1$, but as such, this does not generally imply that $s_1$ and $s_2$ are linear on the whole joint interval $[x_1, x_1 + \Delta_1] \times [x_2, x_2 + \Delta_2]$.

This is not a problem for the proofs, though: in each case, we are only using the fact that $s_1$ and $s_2$ are linear on a joint interval $[x_1, x_1 + \delta] \times [x_2, x_2 + \delta]$ for some small enough $\delta$. Fortunately, this indeed follows from the separate linearities on $[x_1, x_1 + \Delta_1]$ and $[x_2, x_2 + \Delta_2]$ for a choice of $\delta= \frac{\Delta_1+\Delta_2}{2}$. For any debt contract contained in a directed path from $s_1$ or $s_2$ to $T$, if $\hat{e}_{s_1}=x_1 + \Delta_1$ increases the payment on this path by $\eta_1$, and $\hat{e}_{s_2}=x_2 + \Delta_2$ increases the payment on this contract by $\eta_2$, then linearity implies that there are at still at least $\max(\eta_1, \eta_2)$ unpaid liabilities on this edge in the original system with $e_{s_1}=x_1$, $e_{s_2}=x_2$. However, in this case an increase of $\hat{e}_{s_1}=x_1 + \frac{\Delta_1}{2}$ and $\hat{e}_{s_2}=x_2 + \frac{\Delta_2}{2}$ at the source nodes only increases the payment on the contract by $\frac{\eta_1}{2}+\frac{\eta_2}{2} \leq \max(\eta_1, \eta_2)$, so it indeed does not exceed the liability of the edge.

\paragraph*{Concavity.}
The variant of the concavity property formulated in Lemma \ref{lem:conc} is a consequence of a more general property of the network model: given two banks $s$ and $t$, if we consider the funds $e_{s}$ as a variable, then the assets $a_t$ are a concave function of $e_{s}$. This follows from the properties of the network model: if we take two parameters $e_1$ and $e_2$ for the choice of $e_s$ (with $e_1 < e_2$), and assume we add a small $\delta$ amount of extra funds to $s$ in both cases, then these extra funds will follow the same paths through the network until they reach an outgoing liability that has already been paid in full. Due to monotonicity, the payment on each contract with $e_s=e_2$ is at least as much as with $e_s=e_1$; therefore, an increase of $\delta$ at $e_1$ will trigger a larger (or the same) payment increase on each contract than an increase of $\delta$ at $e_2$. This implies that $e_s=e_1+\delta$ increases the value of $a_t$ at least as much as an increase of $e_s=e_2+\delta$. The property also generalizes for arbitrary banks in the system that are not sources or sinks.

Looking at the same property from a shock-based perspective, this means that if a bank $s$ is hit by a shock that removes a specific $\gamma$ amount of funds from this bank (e.g. in the proportional or worst-sum model), then the \textit{loss function} of another bank $t$, i.e. the amount of assets $t$ loses due to this shock, is a convex function of this parameter $\gamma$.

Altogether, we can conclude that the dependency of $a_t$ on the funds $e_s$ of another banks can only happen in a rather restricted fashion in our financial network model: such a dependency is always described by a monotonically increasing, piecewise linear function that is also concave.

\section{No positive swap in the base model} \label{App:B}

In this section, we discuss the details of the proof of Theorem \ref{th:nopos_base} and its corollaries, i.e. that there is no positive swap in the base and proportional models.

The main idea of the proof was already outlined in Section \ref{sec:base}: we consider the open variant of the financial system, and study separate cases based on the relations of the payments before and after the swap. The most straightforward case, i.e. when $p_1' \leq p_1$ and $p_2' \leq p_2$, was already settled in the proof of Lemma \ref{lem:monoton}.

\subsection{Proof of Lemma \ref{lem:bothinc}} \label{sec:l_limits}
We begin with  some technical details omitted from the proof of Lemma \ref{lem:bothinc}, which addresses the case when $p_1' > p_1$ and $p_2' > p_2$.

\renewcommand{\proofname}{Details for the proof of Lemma \ref{lem:bothinc}.}

\begin{proof}
The base idea was already outlined in Section \ref{sec:base}: if we select $e_{s_1}=p_1 + \Delta_1$ and $e_{s_2}=p_2 + \Delta_2$ for a carefully chosen $\Delta_1$ and $\Delta_2$, then this results in $a_{t_1}=p_1 + \Delta_1$ and $a_{t_2}=p_2 + \Delta_2$ due to the linearities in the system. Then these extra assets result in further payments through the swapped edges, thus giving $\Delta_1$ and $\Delta_2$ new assets to $s_1$ and $s_2$, respectively, and therefore producing a larger clearing vector in the original system. 

However, this argument also uses the fact that there are still unpaid liabilities on these contracts; that is, $l_{u_1, v_1} \geq p_1 + \Delta_1$ and $l_{u_2, v_2} \geq p_2 + \Delta_2$. Since the parameter $\delta$ (which defines $\Delta_1$ and $\Delta_2$) can be chosen arbitrarily small, we only need $l_{u_1, v_1} > p_1$ and $l_{u_2, v_2} > p_2$ to satisfy this. It is not hard to see that at least one of these two inequalities certainly holds: having $l_{u_1, v_1} \leq p_1$ and $l_{u_2, v_2} \leq p_2$ would imply $l_{u_1, v_1}+l_{u_2, v_2} \leq p_1 + p_2 < p_1'+p_2'$. However, the same contracts must carry the payments $p_2'$ and $p_1'$ after the swap, so we have $p_2' \leq l_{u_1, v_1}$ and $p_1' \leq l_{u_2, v_2}$, which is a contradiction. Therefore, let us assume w.l.o.g. that $l_{u_1, v_1} > p_1$.

If $l_{u_2, v_2} > p_2$ also holds, then we are ready, so let us assume that $l_{u_2, v_2} = p_2$. We now show that this implies $a_{t_2}' = a_{t_2}$. First of all, due to monotonicity, we have $a_{t_2}' \geq a_{t_2}$. Furthermore, since there is a payment of $p_1'$ from $u_2$ to $v_1$ after swapping in the closed system, we have $p_1' \leq l_{u_2, v_2}$. Recall that we have $a_{t_2}=p_2$ and $a'_{t_2}=p_1'$ from the definition of the open system, so this implies $a'_{t_2} \leq l_{u_2, v_2} = a_{t_2}$.

In order words, $a_{t_2}' = a_{t_2}$ means that all the newly added funds will end up at $t_1$, i.e. both $s_1$ and $s_2$ are $t_1$-linear. However, in this case, setting $e_{s_1}=p_1 + \delta$ would provide $a_{t_1}=p_1 + \delta$, and since $l_{u_1, v_1} \geq p_1 + \delta$ for $\delta$ small enough, this again provides a higher clearing vector in the system.
\end{proof}

\renewcommand{\proofname}{Proof.}

\subsection{Proof of Lemma \ref{lem:hard}}
The more involved part of the proof is to settle the case of Lemma \ref{lem:hard}, i.e. when $p_1' \leq p_1$ and $p_2' > p_2$, but the increased assets at $v_2$ also result in an extra payment for $v_1$ in the system, and this ensures $a_{v_1}' > a_{v_1}$ in the end. Note that we have already seen a similar situation in the semi-positivity example of Figure \ref{fig:semipos}, where $v_2$ received strictly less payment on the swapped edge, but the increased assets of $v_1$ still ensured that $a_{v_2}' = a_{v_2}$. However, in case of a positive swap, we would need $a_{v_1}'$ to be strictly larger than $a_{v_1}$; we show that this is not possible.

We prove Lemma \ref{lem:hard} by starting with a \textit{baseline} open system of $e_{s_1}=p_1'$ and $e_{s_2}=p_2$, i.e. when both banks have the lower amount of funds from the two cases. We then add an extra $p_2'-p_2$ funds to bank $s_2$, and analyze the resulting payment increase in the system; more formally, we analyze the solution $\phi_d$ of the difference system $G_{\text{diff}}$ when injecting these extra funds at $s_2$.

Our analysis will split $\phi_d$ into two different payment configurations (i.e. solutions of artificially defined systems) that sum up to the original payment configuration $\phi_d$. To define this splitting, we first need to find out how many so-called \textit{raw assets} $\mu$ are contributed to $v_1$ due to our increase step. More specifically, we consider a modified variant $G_{\text{raw}}$ of our difference system where we remove all the liabilities of $v_1$, thus making $v_1$ a sink node, and we define the raw assets of $v_1$ as $\mu=a_{v_1} ^{\, (G_{\text{raw}})}$.

Intuitively, $\mu$ does not describe the total amount of new assets that $v_1$ gains due to the increase; similarly to the example of Figure \ref{fig:expansive}, if $v_1$ indirectly receives $\mu$ new assets from $s_2$, then $a_{v_1}$ can increase by more than $\mu$ if $v_1$ is contained in some cycles in the network. Instead, the concept of raw assets aims to capture the fact that from the perspective of $v_1$, increasing $e_{s_2}$ by $p_2'-p_2$ is more or less equivalent to increasing $e_{s_1}$ by $\mu$, since these $\mu$ new funds introduced at $s_1$ would also go through the same cycles.

Let us use $\phi_2$ to denote the payment configuration in the solution of $G_{\text{raw}}$. Note that $\phi_2$ describes all the effects of the increase at $s_2$, except for the behavior of the $\mu$ raw assets after arriving at $v_1$. As such, the remaining effects of the increase are identical to the effects of injecting $\mu$ new funds at $v_1$. More formally, we can form another difference system $G_{\text{rest}}$ by subtracting all the payments in $\phi_2$ from the liabilities on the edges of $G_{\text{diff}}$. We then set the funds of $v_1$ in this system to $e_{v_1}^{\,(G_{\text{rest}})} = \mu$ (and the funds of all other banks $w$ to $e_w \,^{\,(G_{\text{rest}})} = 0$), and denote the payment configuration in the solution of $G_{\text{rest}}$ by $\phi_1$.

Note that this is indeed a partitioning of all our payments: $\phi_1$ and $\phi_2$ sum up to the solution $\phi_d$ of the difference system $G_{\text{diff}}$. Furthermore, recall that $s_1$ has an infinite liability to $v_1$; as such, an alternative (and more convenient) interpretation of $\phi_1$ is that these $\mu$ new funds are instead introduced at the source $s_1$, and $v_1$ receives them from $s_1$. Since the only difference between the two variants is the amount of assets at $s_1$, this two-fold interpretation creates no confusion in our proof.

We now analyze the payment configurations $\phi_1$ and $\phi_2$ in detail. We first begin with a natural observation that if a difference system $G_1$ has less liabilities than another difference system $G_2$, then no bank $w$ can obtain more assets in $G_1$ than in $G_2$.

\begin{lemma} \label{lem:inferior}
Consider two financial networks $G_1$ and $G_2$ on the same banks, such that we have $l_{w_1,w_2}\,\!^{(G_1)} \leq l_{w_1,w_2}\,\!^{(G_2)}$ for any banks $w_1, w_2$. Furthermore assume that there is source node $s$ such that $e_s^{\, (G_1)} = e_s^{\, (G_2)}$, and for any other bank $w \neq s$ we have $e_w^{\, (G_1)} = e_w^{\, (G_2)}=0$. Then for any bank $w$, we have $a_w ^ {\, (G_1)} \leq a_w ^ {\, (G_2)}$.
\end{lemma}

\begin{proof}
Since $s$ is the only node to have any funds in both $G_1$ and $G_2$, each unit of funds exhibits the same behavior (i.e. follows the same path from $s$) in $G_1$ and $G_2$ until it arrives at a node $w_1$ such that $l_{w_1} \,\! ^{(G_1)} < l_{w_1} \,\! ^{(G_2)}$, i.e. a node $w_1$ which has less liabilities in $G_1$ than in $G_2$. At such a node $w_1$, it is possible that in $G_1$ a specific $x$ amount of funds remain at $w_1$, while in $G_2$ these $x$ funds continue traversing the network; however, due to monotonicity, this can only further increase the assets of other banks. As such, any bank $w$ has at least as many assets in $G_2$ as in $G_1$.
\end{proof}

\noindent This already allows us to draw conclusions about the value of $\mu$.

\begin{lemma} \label{lem:mu}
Having $a_{v_1}' > a_{v_1}$ implies that $\mu > p_1-p_1'$.
\end{lemma}

\begin{proof}
Assume first for contradiction that $\mu = p_1-p_1'$. In order to use Lemma \ref{lem:inferior}, we will express both $a_{v_1}$ and $a_{v_1}'$ as the result of taking the assets of $v_1$ in the baseline system first, and then adding $\mu$ extra funds to $s_1$ in a specific difference system.

We can obtain the assets $a_{v_1}$ before swapping in the following way: we start with the baseline system (recall that we have $e_{s_1}=p_1'$, $e_{s_2}=p_2$ here), and we introduce an extra $p_1-p_1'$ funds at $s_1$. Note that the difference system corresponding to this increase has the same liabilities as $G_{\text{diff}}$, it is only the extra funds that are now placed at $s_1$ instead of $s_2$.

On the other hand, after swapping, $v_1$ receives exactly $\mu$ assets in $\phi_2$, which are then reintroduced in $\phi_1$ again; as such, we can obtain $a_{v_1}'$ as the sum of the assets of $v_1$ in the baseline system and the assets of $v_1$ in $\phi_1$. Recall that $\phi_1$ is obtained in the difference system $G_{\text{rest}}$ where the liabilities of $G_{\text{diff}}$ are first further reduced by $\phi_2$.

As the liabilities in $G_{\text{diff}}$ are larger (or the same) than in $G_{\text{rest}}$, Lemma \ref{lem:inferior} implies that we must have $a_{v_1}' \leq a_{v_1}$, which is a contradiction. If we have $\mu < p_1-p_1'$ instead of $\mu = p_1-p_1'$, then the same argument holds, but $\phi_1$ provides even less assets to $a_{v_1}'$ due to monotonicity. As such, for $a_{v_1}' > a_{v_1}$ to hold, we must have $\mu > p_1-p_1'$.
\end{proof}

This implies that there must exist an $\epsilon>0$ such that $\mu = (p_1-p_1') + \epsilon$. As a side note, observe that non-expansivity then implies $p_2'-p_2 \geq \mu = (p_1-p_1') + \epsilon$, which shows that even in this setting, the payments must satisfy $p_1'+p_2' > p_1 + p_2$.

With $\mu = (p_1-p_1') + \epsilon$, let us further partition the configuration $\phi_1$ into two parts, which will intuitively correspond to the first $p_1-p_1'$ funds of the increase at $s_1$, and the remaining $\epsilon$ funds of the increase. More formally, consider the same difference system $G_{\text{rest}}$ again, and let $\phi_{1,1}$ denote the payment configuration in the solution of this system after setting $e_{s_1}=p_1-p_1'$ (as opposed to $e_{s_1}=\mu=(p_1-p_1') + \epsilon$, as in $\phi_1$). Then let us form another difference system $G_{\epsilon}$ by subtracting the payments $\phi_{1,1}$ from the liabilities in $G_{\text{rest}}$ and setting $e_{s_1}=\epsilon$, and denote the payment configuration in the solution of $G_{\epsilon}$ by $\phi_{1,2}$. The definition implies that this is indeed a partitioning of $\phi_1$, i.e. $\phi_1=\phi_{1,1}+\phi_{1,2}$.

The next idea is then to separately analyze how $\phi_2$, $\phi_{1,1}$ and $\phi_{1,2}$ increases the assets of the sink nodes $t_1$ and $t_2$. Recall again that these $3$ configurations together sum up to $\phi$, i.e. they contain all new payments that result from introducing $p_2'-p_2$ extra funds at $s_2$.

\begin{lemma} \label{lem:s1linear}
Bank $s_1$ is $(t_1, t_2)$-linear on $[p_1, p_1+\epsilon]$.
\end{lemma}

\begin{proof}
The configuration $\phi_{1,1}$ provides $(p_1-p_1')$ new funds to $s_1$, thus increasing the funds of $s_1$ to $p_1$. Note that $s_1$ also has the same amount of funds in the system before swapping; in fact, the state before swapping is obtained by adding $(p_1-p_1')$ new funds to $s_1$ in the baseline system. However, $\phi_{1,1}$ is defined in a system where the payments of $\phi_2$ are already subtracted from $G_{\text{diff}}$ first. As such, according to Lemma \ref{lem:inferior}, the funds of $\phi_{1,1}$ can only raise the assets of $t_1$, $t_2$ to $a_{t_1} = p_1$ and $a_{t_2} = p_2$ at most. Hence the remaining increase of assets at the sink nodes (which sums up to $p_1'+p_2'-p_1-p_2$ altogether) must be provided by the payments in $\phi_2$ and $\phi_{1,2}$.

Recall that $G_{\text{raw}}$ has $e_{s_2}=p_2'-p_2$, and nodes $v_1$, $t_1$ and $t_2$ are all sinks in this system; as such, due to non-expansivity, we must have $a_{v_1} + a_{t_1} + a_{t_2} \leq p_2'-p_2$ in the configuration $\phi_2$. Since $\mu$ was defined such that $a_{v_1}=\mu$ in this system, we have $a_{t_1} + a_{t_2} \leq p_2'-p_2 - \mu$ in $\phi_2$.

This implies that $\phi_{1,2}$ must contribute at least $(p_1'+p_2'-p_1-p_2) - (p_2'-p_2 - \mu) = \mu -(p_1 - p_1') = \epsilon$ assets to the two sink nodes. Since $e_{s_1}=\epsilon$ in $\phi_{1,2}$, this means that all the assets in $\phi_{1,2}$ must arrive at the sink nodes, and thus $s_1$ is $(t_1, t_2)$-linear on $[p_1, p_1+\epsilon]$.
\end{proof}

\noindent Due to concavity, this also implies that $s_1$ is $(t_1, t_2)$-linear on the preceding interval.

\begin{corollary}
Bank $s_1$ is $(t_1, t_2)$-linear on $[0, p_1+\epsilon]$.
\end{corollary}

Note that the definition of the open system ensures that all the funds of $s_1$ go directly to $v_1$ in the first step, so this also implies that $v_1$ is also $(t_1, t_2)$-linear.

Finally, this implies that all the new assets at $s_2$ will arrive at the sink nodes, too.

\begin{lemma}
Bank $s_2$ is $(t_1, t_2)$-linear on $[p_2, p_2']$.
\end{lemma}

\begin{proof}
The proof of Lemma \ref{lem:s1linear} shows that the total increase of $(p_1'+p_2'-p_1-p_2)$ can only be obtained if $\phi_2$ contributes $(p_2'-p_2) - \mu$ assets directly to $t_1$ and $t_2$, and the remaining $\mu$ assets go to bank $v_1$. Since $v_1$ is $(t_1, t_2)$-linear, this means that all the $(p_2'-p_2)$ new funds of $s_2$ end up at $t_1$ or $t_2$.
\end{proof}

Once again, concavity implies that $s_2$ is $(t_1, t_2)$-linear on $[0, p_2']$.

Given these linearity properties, we can finish our proof of Theorem \ref{th:nopos_base} along the same lines as the proof of Lemma \ref{lem:bothinc}: we can select a small increase of $\Delta_1$ and $\Delta_2$ that results in a larger clearing vector in the closed system before swapping. The arguments of Appendix \ref{sec:l_limits} on the liabilities can be applied in this case, too.

\subsection{Proportional shocks}

Finally, we use a similar line of thought to prove Theorem \ref{th:nopos_prop}, i.e. that there is also no positive swap in the proportional shock model. As already discussed in Section \ref{sec:prop}, the only technically involved case in the proof is when a swap provides $f'_{v_1}(\lambda_1)>f_{v_1}(\lambda_1)$ and $f'_{v_2}(\lambda_1)=f_{v_2}(\lambda_1)$ for $\lambda_1$, and $f'_{v_1}(\lambda_2)=f_{v_1}(\lambda_2)$ and $f'_{v_2}(\lambda_2)>f_{v_2}(\lambda_2)$ for $\lambda_2$. Let us assume w.l.o.g. that $\lambda_1<\lambda_2$.

\renewcommand{\proofname}{Details for the proof of Theorem \ref{th:nopos_prop}.}

\begin{proof}
The main idea of the proof is to show that while semi-positive swaps are actually possible in the base model (recall the example from Figure \ref{fig:semipos}), they can only happen in rather restricted cases.  We can prove this in an analogous way to the proof of Theorem \ref{th:nopos_prop}.

In particular, we cannot have $p_1' \leq p_1$ and $p_2' \leq p_2$, because then monotonicity would imply that the swap is not beneficial for either of the banks. As such, one of the acting nodes receives strictly more payment on the swapped contract after swapping; assume again w.l.o.g. that this is $v_2$, i.e. $p_2'>p_2$. Note that we cannot have $p_1'>p_1$ simultaneously to this, since this would imply $a_{v_1}' > a_{v_1}$, so the swap would be positive. Hence the only way to have a semi-positive swap is to have $p_1' \leq p_1$ and $p_2' > p_2$, resulting in $a_{v_1}' = a_{v_1}$ and $a_{v_2}' > a_{v_2}$.

Once again, we can define a difference system $G_{\text{diff}}$ when increasing $e_{s_2}$ from $p_2$ to $p_2'$ in the baseline system, and then a system $G_{\text{raw}}$ to identify the amount of raw assets $\mu$ that are sent from $s_2$ to $v_1$ in this increase. We can then use the same argument as in the proof of Lemma \ref{lem:mu} to show that we must have $\mu \geq (p_1 - p_1')$; otherwise, $\mu < (p_1 - p_1')$ would imply $a_{v_1}' < a_{v_1}$.

Note that we also cannot have $\mu > (p_1 - p_1')$: in this case, we could use the proof of Lemma \ref{lem:s1linear} to show that the last $\mu - (p_1 - p_1')$ units of funds would have to contribute $\mu - (p_1 - p_1')$ assets to $t_1$ and $t_2$; this would again make $v_1$ $(t_1,t_2)$-linear, thus leading to the same contradiction as in case of positive swaps. As such, a semi-positive swap can only happen if we have $\mu = p_1 - p_1'$ exactly. This implies that in the solution $\phi$ of $G_{\text{raw}}$, the remaining $p_2' - p_2 - \mu = (p_1' + p_2') - (p_1 + p_2)$ new funds must all arrive at the sink nodes to trigger the desired increase of $(p_1' + p_2') - (p_1 + p_2)$.

As such, semi-positivity does not imply that $s_2$ is $(t_1,t_2)$-linear, but implies something almost as good: that the new funds introduced in $s_2$ can be partitioned into a part that is $(t_1,t_2)$-linear, and into a part that directly ends up in $v_1$ (or in other words: $s_2$ is $(t_1, t_2, v_1)$-linear in $G_{\text{raw}}$ where $v_1$ is artificially turned into a sink node). This will allow for the same kind of proof technique as before, because in the closed version of the system, this still means that essentially any new funds at $v_2$ will arrive at either $v_1$ or $v_2$.

Now let us consider proportional shocks again. Note that all the linearity properties we have discussed (including $(t_1, t_2, v_1)$-linearity in $G_{\text{raw}}$) are preserved if we scale down the funds of each bank in the system proportionally. This implies that these linearites for bank $v_1$ which follow from semi-positivity in case of $\lambda_1$ also carry over to the case of $\lambda_2$. As such, for the larger parameter $\lambda_2$, we have such a linearity in both directions: $v_1$ is $(t_1,t_2)$-linear apart from some assets that go directly to $v_2$, and $v_2$ is $(t_1,t_2)$-linear apart from some assets that go directly to $v_1$.

We again use the technique from the proof of Lemma \ref{lem:bothinc} to show that this is a contradiction in the case when the shock size is $\lambda_2$. Note that in the closed system (before the swap), the assets provided to $t_2$ and the raw assets going directly to $v_2$ will both end up at $v_2$, so we can cover them with the same coefficient. That is, we can just consider the two constants $\alpha_1$ and $\beta_1$ that fulfill the following roles for a small increase $\delta$ in $e_{s_1}$: (i) the assets of $t_1$ increase by $\alpha_1 \cdot \delta$, and (ii) the sum of the asset increase at $t_2$ and the new raw assets provided to $v_2$ in the open system is altogether $\beta_1 \cdot \delta$. Given these constants, we can use the same method to select appropriate values $\Delta_1$ and $\Delta_2$ that provide a slightly larger clearing vector in the system.

Finally, as in Appendix \ref{sec:l_limits} before, let us discuss the edge cases when the swapped contracts do not have any remaining liabilities to relay such an increase. If we have $p_2=p_1'=l_{u_2, v_2}$, then the increase at $s_1$ cannot contribute any assets to $t_2$. In this case, if we have $\beta_1>0$ nonetheless (there are raw assets going directly to $v_2$), then the same proof works for without any modification. Otherwise we have $\alpha_1=1$, and hence a choice of $\Delta_1=\delta$ and $\Delta_2=0$ suffices, assuming that $p_1<l_{u_1, v_1}$. Finally, if both swapped contracts are fully paid (this can indeed happen in this case, just consider Figure \ref{fig:semipos} with the modification $l_{u_1, v_1}=\frac{1}{2}$), then the increase at either of the source nodes cannot contribute to $t_1$ or $t_2$, and hence an increase of $\delta$ at $s_1$ (or $s_2$, respectively) must provide $\delta$ raw assets going to $v_2$ (and $v_1$, respectively). In this case, there is no need for the swapped contracts; a choice of $\Delta_1=\Delta_2=\delta$ provides a larger clearing vector, even in the open version of the system.
\end{proof}

\renewcommand{\proofname}{Proof.}

\section{Adapting our results to the worst-sum model} \label{App:D}

We now discuss how to adapt the proofs in Section \ref{sec:worstcase} to the worst-sum shock model. Since the two models both assume a worst case, the base ideas of the proofs will remain similar, and we mostly only need to execute some technical changes.

We first discuss an important property of the worst-sum model that will be helpful in the proofs of our theorems.

\begin{lemma} \label{lem:nosplit}
Given a worst-sum shock of a specific size $\rho$ for any node $v$, there is at most $1$ bank that loses only a fraction of its funds, i.e. each other bank $w$ either has $e_w=0$ after the shock, or it does not lose any funds in the shock.

If there are multiple possible worst-sum shocks of size $\rho$ for $v$, then there exists at least one of them that fulfills this property.
\end{lemma}

\begin{proof}
Recall from the general properties of our financial systems that the assets $a_v$ can only depend in a specific way on the funds $e_w$ of another bank $w$: as we decrease the funds of $w$ from an original value $e_w$ to $0$, $a_v$ is a concave, monotonically decreasing piecewise linear function.

Now assume that two banks $w_1$ and $w_2$ are both partially hit by the shock, losing $\rho_1$ and $\rho_2$ funds respectively. Consider the slope (first derivative) $\varphi_1$ and $\varphi_2$ of the loss functions at the two specific points (if the points are the breakpoints of the piecewise linear functions, we can simply take the average of the slopes of the line segments before and after the breakpoints).

Assume w.l.o.g. that $\varphi_1 \leq \varphi_2$. Note that due to the convexity of both loss functions, this implies that the loss function with respect to $e_{w_1}$ has a slope of at most $\varphi_1$ for each $\hat{\rho}_1 < \rho_1$, and the loss function with respect to $e_{w_2}$ has a slope of at least $\varphi_2$ for each $\hat{\rho}_2 > \rho_2$. This means that we can redistribute the shock, taking away more funds from $w_2$ and less funds from $w_1$, and any such change can only reduce (or at least not increase) the final value of $a_v$. If we do this until we reach $\rho_1=0$ or $\rho_2=e_{w_2}$, we obtain a new shock with $a_v$ at most as much as before, but a strictly smaller number of nodes are hit partially by the shock. Executing this step repeatedly gives a shock distribution where at most one bank is hit partially.

We note that formally, the proof is slightly more complicated, since the two loss functions of $v$ (as a function of $e_{w_1}$ and as a function of $e_{w_2}$) might also depend on each other. Assuming that both $w_1$ and $w_2$ are partially hit by the shock (losing $\rho_1$ and $\rho_2$ funds), we need to consider (i) the loss of $v$ as a function of $e_{w_1}$ for this fixed value $\rho_2$, and (i) the loss of $v$ as a function of $e_{w_2}$ for this fixed value $\rho_1$. We can once again consider the bank $w_i$ where the loss function is less steep at the chosen point (e.g. $w_1$), and reduce $\rho_1$ until we reach $\rho_1=0$ or $\rho_1=e_{w_2}-\rho_2$. If we keep $\rho_2$ fixed during this operation, then the loss function with respect to $w_1$ does not change, so our arguments about steepness still holds.

Then we consider the second step of increasing $\rho_2$ with regard to this new reduced $\rho_1$. In this second step, there are less funds in the system since $e_{w_1}$ was reduced; this can only make the loss function with respect to $w_2$ steeper at each point, since there might be more banks now that cannot fulfill their obligations. As such, the second step of increasing $\rho_2$ is still valid, since the function slope at every higher point is still larger than it was at $\rho_2$ originally.
\end{proof}

\noindent Now we can consider the specific results from Section \ref{sec:worstcase}.

\begin{theorem}
Lemmas \ref{lem:badforw} and \ref{lem:badforu} also hold in the worst-sum model.
\end{theorem}

\begin{proof}
For Lemma \ref{lem:badforw}, consider the network topology of Figure \ref{fig:badforw} in the worst-sum setting, but with some slight modifications: we add another outgoing debt of weight $2$ from $s_1$ to a new sink node $t$ to make the situation of the two source nodes symmetrical. We also set $e_{s_1}=8$ and $e_{s_2}=8$ in this new system.

One can observe that the shock functions of $v_1$ and $v_2$ in this system are the same as the top-row worst total shock function in Figure \ref{fig:motive}, but scaled to twice the original width along the $x$ axis; this holds both before and after the swap, so the swap is indeed still positive for $v_1$ and $v_2$.

On the other hand, the original shock function of $w$ is the same as the bottom-row worst total shock function in Figure \ref{fig:motive}, now scaled to twice its original size along both axes. After the swap, the shock function of $w$ becomes the piecewise linear function consisting of the segments $(0,4)-(4,4)$, $(4,4)-(8,1)$, $(8,1)-(12,1)$ and $(12,1)-(16,0)$, so it indeed satisfies $f_w'<f_w$.

The construction for Lemma \ref{lem:badforu} is more difficult to adapt to this setting. Note that the swap in the original version of this network is not positive in the worst-sum case, since a shock of size $\rho \in (0, \frac{2}{3})$ is harmless to $v_2$ before the swap ($s_2$ can still fulfill its obligations), but it reduces the assets of $v_2$ after the swap. This value $\frac{2}{3}$ turns out to be a crucial threshold in this system, since this is the amount of loss after which $s_2$ starts making less payments. As such, for our proof, let us consider a modified version of Figure \ref{fig:badforu} where we also increase the assets of $s_1$ by this amount, i.e. we select $e_{s_1}=2+\frac{2}{3}$.

Let us first consider the shock functions before the swap in this new system. Bank $v_1$ loses no funds until $\rho=\frac{2}{3}$, then all funds until $\rho=2+\frac{2}{3}$, giving the piecewise linear function $(0,2)-(\frac{2}{3},2)-(2+\frac{2}{3},0)$. Bank $v_2$ loses no funds until $\rho=\frac{2}{3}$, and then its funds drop linearly until $\rho=2$; at this point, it only receives the $\frac{2}{3}$ units of money that are indirectly coming from $s_1$. From this point, its assets remain fixed for another $\frac{2}{3}$ units of shock (while the extra funds of $s_1$ are depleted), and then they drop linearly to $0$. This defines the function $(0,2)-(\frac{2}{3},2)-(2,\frac{2}{3})-(2+\frac{2}{3},\frac{2}{3})-(4+\frac{2}{3},0)$. The assets of $u_2$ are simply half of the assets of $v_2$ for any $\rho$.

Now consider the system after the swap. This system is more complex to analyze since all $3$ banks indirectly receive assets of $3$ different kinds: those originating from $e_{s_1}$, those originating from $e_{s_2}$, and the ones coming on the backward edge (which, as one can compute, carries a payment of $\frac{1}{5} \cdot \left( e_{s_2} + \min(e_{s_1}, 2) \right)$ from the point where $s_2$ goes into default). Once again it holds that none of our $3$ bank lose assets until $\rho \leq \frac{2}{3}$. One can observe that after this point, the shock is worse for all $3$ banks if we start depleting the funds of $s_2$ first; intuitively, this is because the first $\frac{2}{3}$ units of reduction at $s_2$ also decrease the payment coming on the backward edge, while the first $\frac{2}{3}$ units of reduction at at $s_1$ do not.

As such, the shock functions are as follows. Both $v_1$ and $v_2$ lose no funds until $\rho=\frac{2}{3}$, and when $s_2$ has lost all of its funds, they both have assets of $\frac{6}{5}$. From this point, the extra funds of $s_1$ are removed for another $\frac{2}{3}$ units, and then the remaining assets are lost linearly until $\rho=4+\frac{2}{3}$. This defines the shock function $(0,2)-(\frac{2}{3},2)-(2,\frac{6}{5})-(2+\frac{2}{3},\frac{6}{5})-(4+\frac{2}{3},0)$ for both of the acting nodes. Bank $u_2$, on the other hand, has assets of $\frac{1}{5}$ when the funds of $s_2$ are all removed, so it has the shock function $(0,2)-(\frac{2}{3},2)-(2,\frac{1}{5})-(2+\frac{2}{3},\frac{1}{5})-(4+\frac{2}{3},0)$. This is indeed a strictly better function for $v_1$ and $v_2$, and a strictly worse one for $u_2$.
\end{proof}

\noindent We continue with our results that study the worst-case models from a computational perspective.

\begin{theorem}
Theorem \ref{th:hardness} also holds in the worst-sum model.
\end{theorem}

\begin{proof}
We can essentially apply the same reduction idea as in the worst-set model, with a minor modification: let $D$ denote the maximal degree in the input graph $H$, and instead of setting $e_s=\text{deg}_z$ for a vertex $z$, we set $e_s=D$ for each source bank uniformly.

We then consider the problem with a shock of size $\rho=k \cdot D$. Due to Lemma \ref{lem:nosplit}, we know that the worst shock of size $\rho$ will hit exactly $k$ distinct sources in this network, and remove their funds completely. Hence the loss of assets at $v$ is again exactly the number of edges covered by the densest $k$ subgraph, so we can apply the same reduction as before.
\end{proof}

The following Corollary \ref{cor:constsmall} is not straightforward to adapt to this case, since the shock size $\rho$ does not directly limit the number of banks that are hit by the shock: even with Lemma \ref{lem:nosplit}, it is possible that the shock is distributed among a high number of banks if they all have significantly less funds than $\rho$. However, the corollary does carry over if we exclude such cases, e.g. if it holds for some constant $k$ that all banks $u$ in the system have either $e_u \geq \frac{\rho}{k}$ or $e_u=0$.

Devising a dynamic programming algorithm for the special case of tree networks is also a more complex problem in the worst-sum case. Note that similarly to Theorem \ref{th:treedp}, if have we two debtors $w_1$ and $w_2$ with already known shock functions, then we can still compute $f_u(\rho)$ for a specific value $\rho$. Assume that the worst shock of size $\rho$ can be generated by two shocks $f_{w_1}(\rho_1)$ and $f_{w_2}(\rho_2)$ for some $\rho_1+\rho_2=\rho$. One can observe that either $\rho_1$ or $\rho_2$ must be a breakpoint of the corresponding piecewise linear function, since otherwise we can redistribute some units of shock to the steeper function to obtain a larger total shock (in case of identical slopes, we can also do this until $\rho_1$ or $\rho_2$ becomes a breakpoint). As such, in order to find $f_u(\rho)$, it suffices to consider the sums $f_{w_1}(\rho_1)+f_{w_2}(\rho-\rho_1)$ for all breakpoints $\rho_1$ of $f_{w_1}$ and the sums $f_{w_2}(\rho_2)+f_{w_1}(\rho-\rho_2)$ for all breakpoints $\rho_2$ of $f_{w_2}$, and select the smallest one among them.

The problem with this method, however, is that even though the worst-sum shock functions are piecewise linear, it is non-trivial to prove that they only consist of polynomially many segments; without this, even the representation of these functions becomes problematic. Intuitively, this is because when we merge $f_{w_1}$ and $f_{w_2}$ at $u$, then some breakpoints of the new $f_u$ might not directly relate to breakpoints of $f_{w_1}$ and $f_{w_2}$. For example, the segment following $f_u(\rho)$ may not be the entire next segment of $f_{w_1}(\rho_1)$ or $f_{w_2}(\rho_2)$; it can happen that somewhere within this segment, it becomes more optimal to switch back to a previous breakpoint $\hat{\rho}_1<\rho_1$ of $f_{w_1}$ and redistribute the remaining assets to $f_{w_2}$ instead. Since the number of such intersection points is not straightforward to upper bound, we leave it to future work to conduct a more thorough survey of this shock function from a computational perspective. 

However, we note that if the shock function of every bank consists of e.g. $O(n)$ breakpoints only, then the dynamic programming approach can already be adapted to this setting. If we know the value of $f_u(\rho)$ and the corresponding $\rho_1$ and $\rho_2$, then $f_u$ continues after $\rho$ with the steeper one of the functions $f_{w_1}(\rho_1)$ and $f_{w_2}(\rho_2)$ at the given point. The next breakpoint in $f_u$ is either obtained from the endpoint $\hat{\rho}$ of the respective segment in $f_{w_1}$ or $f_{w_2}$, or it can happen even before $\hat{\rho}$ if it is a previously described intersection point (obtained by redistributing the first $\rho$ units of shock). 
To find the earliest such intersection point, we need to revisit all previous breakpoints of $f_{w_1}$ and $f_{w_2}$: e.g. for a breakpoint $\hat{\rho}_1<\rho_1$ of $f_{w_1}$, we must check whether the shock combined from $f_{w_1}(\hat{\rho}_1)$ and $f_{w_2}(\hat{\rho}-\hat{\rho}_1)$ intersects our segment on the interval $[\rho, \hat{\rho}]$.

On the other hand, the construction idea for a positive swap in tree networks can be adapted to the worst-sum model with some simple modifications.

\begin{theorem}
Theorem \ref{th:treepos} also holds in the worst-sum model.
\end{theorem}

\begin{proof}
We can use the same construction idea as is in the original proof; we only need to use a different version of the $d$-boolean gadgets in this model. Consider the parameters $d_i$ of all the boolean gadgets we use in the construction, and let us select a large constant $D$ that satisfies $D>d_i$ for all the values $d_i$. Then our new $d$-boolean gadget will be a rather simple construction: it will consist of a single node $w$ with $e_w=D \cdot d$, and an outgoing debt of $d$ towards the desired acting node.

With this choice of parameters, this single bank essentially implements the same behavior as the original gadget in the worst-set model, with $D$ essentially becoming the new ``unit'' of loss. More specifically, given a multiset of integers $S$ that contains the parameters of boolean gadgets attached to an acting node, for any $\rho = h \cdot D$ for an integer $h$ (i.e. when $\rho$ is a multiple of $D$), the worst possible loss for an acting node can be obtained as the largest sum in $S$ that still does not exceed $h$. As for the values $\rho$ that are not multiples of $D$: if the multiset $S$ allow us to select a subset that sums up to an integer $h_1$, and the largest integer below $h_1$ that can be formed from $S$ is $h_2$, then the loss between $h_2 \cdot D$ and $h_1 \cdot D$ will be described by the segments $(h_2 \cdot D,\, h_2)-(h_1 \cdot D - (h_1-h_2) ,\,h_2)$ and $(h_1 \cdot D - (h_1-h_2) ,\,h_2)-(h_1 \cdot D,\, h_1)$.

Hence if we select the same gadget parameters as in the original construction, then the shock function we receive will essentially be a continuous version of the worst-set shock function: (i) the worst-set function is first scaled $D$-wise wider along to the horizontal axis, (ii) each discrete point is turned into a horizontal line that goes until the $x$ coordinate of the next discrete point (forming a decreasing ``step function''), and (iii) then each vertical drop at the discrete points is replaced by a decrease of slope $-1$ that ends at the given point. As such, the same proof can be applied as in the original case.

The $1$-fix gadget can also be easily adapted: we just create a single node $u$ with more than $K+1$ funds (where $K$ is the upper limit on the shock size), and an outgoing debt of $1$ from $u$.
\end{proof}

\noindent Finally, we revisit our results on portfolio swapping and debt reorganization.

\begin{theorem}
Theorems \ref{th:portfolio} and \ref{th:reorg} also hold in the worst-sum model.
\end{theorem}

\begin{proof}
Theorem \ref{th:portfolio} can easily be adapted to the worst-sum case with the same construction as in the original proof. Note that a combination of $a_{v_1}=39 \cdot r_{s_1} + 37 \cdot r_{s_2}$ (or vice versa) means that the worst-sum shock function of the acting nodes consists of $2$ linear segments: $(0,76)-(72,37)-(144,0)$. The portfolio swapping in the example improves this to a single segment $(0,76)-(144,0)$. On the other hand, since any simple swap produces a worse combination $x_1 \cdot r_{s_1} + x_2 \cdot r_{s_2}$ where either $x_1>39$ or $x_2>39$, they all create a strictly worse shock function than the original one.

For Theorem \ref{th:reorg}, there is once again no need to change our construction from the original proof. The original configuration describes a shock function composed of $(0,3)-(3,1)-(6,0)$, while the reorganization results in a single segment $(0,3)-(9,0)$. Since any single swap results in a symmetrical variant of the original shock function for one of the acting nodes, there is again no positive swap in the network.
\end{proof}

\end{document}